\documentclass{article}
\usepackage[utf8]{inputenc}
\usepackage[toc,page]{appendix}
\usepackage[margin=1in]{geometry}
\usepackage[bookmarks, colorlinks=true, plainpages = false, citecolor = blue,linkcolor=red,urlcolor = blue, filecolor = blue,pagebackref]{hyperref}
\usepackage{url}
\usepackage{amsmath,amsfonts,amsthm,amssymb,bm, verbatim,dsfont,mathtools}
\usepackage{color,graphicx,appendix}
\usepackage{subfigure}
\usepackage{etoolbox}
\usepackage{tikz}
\usepackage{xr,xspace}
\usepackage{todonotes}
\usepackage{paralist}
\usepackage{caption,soul}
\usepackage[ruled,vlined,linesnumbered]{algorithm2e}
\usepackage{appendix}
\makeatletter

\newtheorem{theorem}{Theorem}
\newtheorem{lemma}[theorem]{Lemma}

\theoremstyle{definition}
\newtheorem{definition}[theorem]{Definition}

\newtheorem{claim}[theorem]{Claim}

\newtheorem{remark}[theorem]{Remark}

\newtheorem{fact}[theorem]{Fact}

\usepackage{xspace,prettyref}

\newcommand{\diverge}{\to\infty}

\newcommand{\reals}{{\mathbb{R}}}

\definecolor{ForestGreen}{rgb}{0, 0.5, 0} 

\newcommand{\red}{\color{red}}
\newcommand{\blue}{\color{blue}}
\newcommand{\nb}[1]{{\sf\blue[#1]}}
\newcommand{\nbr}[1]{{\sf\red[#1]}}

\newcommand{\prob}[1]{ \mathbb{P}\left\{ #1 \right\} }

\newcommand{\iid}{i.i.d.\xspace}
\newrefformat{eq}{(\ref{#1})}
\newrefformat{chap}{Chapter~\ref{#1}}
\newrefformat{sec}{Section~\ref{#1}}
\newrefformat{alg}{Algorithm~\ref{#1}}
\newrefformat{fig}{Fig.~\ref{#1}}
\newrefformat{tab}{Table~\ref{#1}}
\newrefformat{rmk}{Remark~\ref{#1}}
\newrefformat{clm}{Claim~\ref{#1}}
\newrefformat{def}{Definition~\ref{#1}}
\newrefformat{cor}{Corollary~\ref{#1}}
\newrefformat{lmm}{Lemma~\ref{#1}}
\newrefformat{prop}{Proposition~\ref{#1}}
\newrefformat{app}{Appendix~\ref{#1}}
\newrefformat{hyp}{Hypothesis~\ref{#1}}
\newrefformat{thm}{Theorem~\ref{#1}}
\newrefformat{ass}{Assumption~\ref{#1}}
\newrefformat{problem}{Problem~\ref{#1}}

\newcommand{\pth}[1]{\left( #1 \right)}
\newcommand{\qth}[1]{\left[ #1 \right]}
\newcommand{\sth}[1]{\left\{ #1 \right\}}

\newcommand{\abth}[1]{\left | #1 \right |}

\newcommand{\indc}[1]{{\mathbf{1}_{\left\{{#1}\right\}}}}

\newcommand{\calA}{{\mathcal{A}}}

\newcommand{\calE}{{\mathcal{E}}}

\newcommand{\calN}{{\mathcal{N}}}

\newcommand{\calP}{{\mathcal{P}}}

\newcommand{\calR}{{\mathcal{R}}}

\newcommand{\calW}{{\mathcal{W}}}

\newcommand{\weight}{\mathsf{w}}

\renewcommand{\hat}{\widehat}

\newcommand{\pre}{\mathsf{PRE}}
\newcommand{\pos}{\mathsf{POST}}

\title{Spike-Based Winner-Take-All Computation: \\
Fundamental Limits and Order-Optimal Circuits
} 
\author{
Lili Su \\
Computer Science \& Artificial Intelligence Laboratory\\
Massachusetts Institute of Technology\\
{lilisu@mit.edu}
\and 
Chia-Jung Chang\\
Brain and Cognitive Sciences\\
Massachusetts Institute of Technology\\
{chiajung@mit.edu}
\and 
Nancy Lynch\\
Computer Science \& Artificial Intelligence Laboratory\\
Massachusetts Institute of Technology\\
{lynch@csail.mit.edu}
}

\date{\today}

\begin{document}

\maketitle

\begin{abstract}
Winner-Take-All (WTA) refers to the neural operation that selects a (typically small) group of 
neurons from a large neuron pool. It is conjectured to underlie many of the brain's fundamental computational abilities. However, not much is known about the robustness of a spike-based WTA network to the inherent randomness of the input spike trains. In this work, we consider a spike-based $k$--WTA model wherein $n$ randomly generated input spike trains compete with each other based on their underlying statistics, and $k$ winners are supposed to be selected. 
We slot the time evenly with each time slot of length  
$1\, ms$, and model the $n$ input spike trains as $n$ independent Bernoulli processes. 
The Bernoulli process is a good approximation of the popular Poisson process but is more biologically relevant as it takes the refractory periods into account. 
Due to the randomness in the input spike trains, no circuits 
can guarantee to successfully select the correct winners in finite time. We focus on analytically characterizing the minimal amount of time needed so that a target minimax decision accuracy (success probability) can be reached. 

%

We first derive an information-theoretic lower bound on the decision time. 
We show that to have a (minimax) decision error $\le \delta$ (where $\delta \in (0,1)$), the computation time of any WTA circuit is at least  
\[
((1-\delta) \log(k(n -k)+1) -1)T_{\calR},
\]
where 
$T_{\calR}$ is a difficulty parameter of a WTA task that is independent of $\delta$, $n$, and $k$. 
We then design a simple WTA circuit whose decision time is 
\[
O\pth{\pth{\log \pth{\frac{1}{\delta}}+\log k(n-k)}T_{\calR}}.
\]
It turns out that for any fixed $\delta \in (0,1)$, this decision time is order-optimal in terms of its scaling in $n$, $k$, and $T_{\calR}$. 
%

\end{abstract}

\section{Introduction}
Humans and animals can form a stable perception and make robust judgments under 
ambiguous conditions. For example, we can easily recognize a dog in a picture regardless of its posture, hair color, and whether it stands in the shadow or 
is occluded by other objects. One fundamental feature of brain computation is its robustness to the randomness 
introduced at different stages, such as sensory representations \cite{kinoshita2001neural, hubel1959receptive}, feature integration \cite{kourtzi2003integration, majaj2007motion}, decision formation \cite{platt1999neural, shadlen2001neural}, and motor planning \cite{harris1998signal, li2015motor}. It has been shown that neurons encode information in a stochastic manner in the brain \cite{baddeley1997responses, kara2000low, maimon2009beyond, ferrari2018simple}; even when the exact same sensory stimulus is presented or when the same kinematics are achieved, no  deterministic patterns in the spike trains exist. Facing environmental ambiguity, humans and animals adaptively refine 
their behaviors by incorporating prior knowledge with their current sensory measurements \cite{faisal2008noise, knill2004bayesian, stocker2006noise, ernst2002humans, kording2004bayesian}. Nevertheless, it remains relatively unclear how neurons carry out robust computation facing ambiguity. 
Sparse coding is a common strategy in brain computation; to encode a task-relevant variable, often only a small group of neurons from a large neuron pool are activated \cite{olshausen2004sparse, perez2002oscillations, hromadka2008sparse, quiroga2008sparse, karlsson2008network, redgrave1999basal}. Understanding the underlying neuron selection mechanism is highly challenging. 

Winner-Take-All (WTA) is a hypothesized 
mechanism  to select proper neurons from a competitive network of neurons, 
and is conjectured to be a fundamental primitive of 
cognitive functions such as attention and object recognition \cite{riesenhuber1999hierarchical, itti1998model, yuille1998handbook, maass2000computational}. Among these studies, it is commonly assumed that neurons transmit information with a continuous variable such as the firing rate. This assumption, however, ignores how temporal coding may additionally contribute to cortical computations. For example, some neurons in the auditory cortex will respond to auditory events with bursts at a fixed latency \cite{gerstner1996neuronal, nelken2004processing}. This phase-locking property is also observed in the hippocampus as well as the prefrontal cortex \cite{siapas2005prefrontal, hahn2006phase, buzsaki1995temporal}. Another feature that has been neglected in a rate-based model is the inherent noise in the inputs. Although some studies used additive Gaussian noise \cite{kriener2017fast, li2013class, lee1999attention, rougier2006emergence} to account for input randomness, such WTA circuits are very sensitive to noise and could not successfully select even a single winner unless extra robustness strategy such as an additional nonlinearity is introduced 
into the dynamics \cite{kriener2017fast}. Last but not least, neurons have a refractory period, which prevents spikes from back propagating in axons \cite{berry1998refractoriness}, and such a feature is usually neglected in the rate-based models. 
In contrast, a spike-based model may capture these neglected features. Nevertheless, how WTA computation can be implemented and its algorithmic characterization remains relatively under-explored.

In this paper, we study a spike-based $k$-WTA model wherein $n$ randomly generated input spike trains are competing with each other with their underlying statistics, and the true winners are the $k$ input spike trains whose underlying statistics are higher than others. 
More precisely, we slot the time evenly with each time slot of length $1\, ms$. We assume that these $n$ input spike trains are generated by $n$ independent Bernoulli processes with different rates. 
An abstract example 
is depicted in Figure \ref{fig:WTA abstract}. 
We use Bernoulli processes to capture the randomness in the input spike trains 
rather than using the popular Poisson processes because a Bernoulli process can be viewed as the time-slotted version of a refractory-period-modified Poisson process; it is well-known that due to the existence of refractory periods, a neuron cannot spike twice within $1\, ms$. 

We focus on analytically characterizing the minimal  amount  of  time  needed  so  that  a  target  minimax  decision  accuracy  (success probability) can be reached.
We first derive a lower bound on the decision time for a given decision accuracy. 
We show that no WTA circuits can have a computation time strictly less than 
\begin{align}
    \label{eq: lower bound}
((1-\delta) \log(k(n -k)+1) -1)T_{\calR},
\end{align}
where $T_{\calR}$ is a parameter defining the difficulty to distinguish between two spike trains with different statistics in a WTA task, $n$ is the number of input spike trains, $k$ is the number of winners, and $\delta$ is the given target decision accuracy. In many practical settings we care about the sparse coding region where $k\ll n$. 
Our lower-bound is obtained by an information-theoretic argument, and holds for all WTA circuits without restricting their circuit architectures and their adopted activation functions. Throughout this paper, we are interested in the decision time's scaling in $n$, $k$, and $T_{\calR}$, while treating $\delta \in (0,1)$ as a fixed small constant. 
Not surprisingly, the above lower bound grows with the network size $n$ when other parameters are fixed. This is because the larger $n$, the noisier the WTA competition. 
Similarly, when $n$ and $k$ are fixed, the easier to distinguish two spike trains with different statistics (i.e., the smaller $T_{\calR}$), the shorter the necessary decision time is. 
%
\begin{figure}
\centering
\includegraphics[width=14cm]{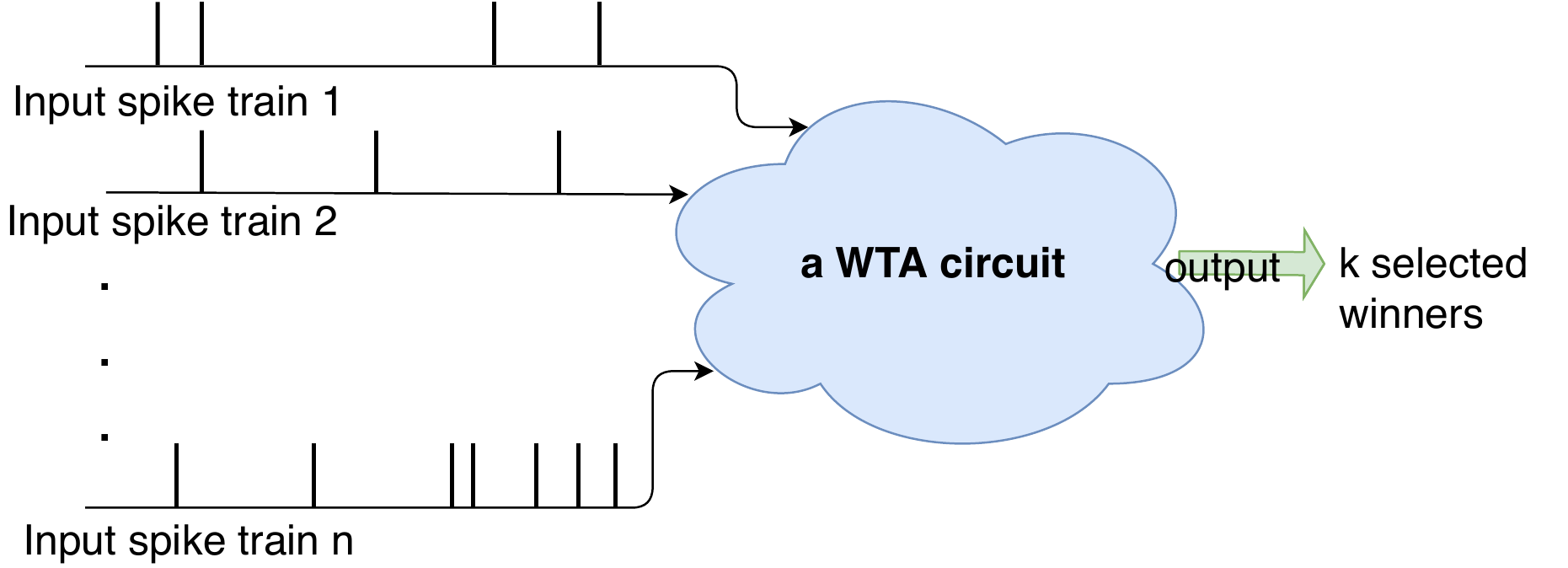}
  \caption{In this figure, $n$ randomly generated input spike trains are fed to a WTA circuit as the circuit input. 
  Clearly, no deterministic patterns can be read off from these spike trains. Here, we do not specify the output of a WTA circuit because 
the detailed specifications of the circuits' outputs might vary with the corresponding applications. 
%
  }
\label{fig:WTA abstract}
\end{figure}

We construct a simple circuit whose decision time is 
\[
O\pth{\pth{\log \pth{\frac{1}{\delta}}+\log k(n-k)}T_{\calR}}.
\]
It turns out that for any fixed $\delta \in (0,1)$, this decision time is order-optimal in terms of its scaling in $n$, $k$, and $T_{\calR}$, i.e., its decision time matches the lower-bound in \eqref{eq: lower bound} up to a constant multiplicative factor. 
In our circuit, each output neuron is a thresholded accumulator unit whose 
threshold is determined by $\delta$, $k$, $n$, and $T_{\calR}$, and the circuit's output is the first group of $k$ output neurons that spike in the same time.  
The typical dynamics under our circuit are: the number of output neurons that spike simultaneously (i.e., spike at the same time) is monotonically increasing until exactly $k$ output neurons spike simultaneously. The simultaneous spikes of these $k$ output neurons cause strong inhibition of other output neurons; in particular, no other output neuron can spike within a sufficiently long period $\Omega\pth{\pth{\log \pth{\frac{1}{\delta}}+\log k(n-k)}T_{\calR}}$. 


In addition, our results 
also give a set of testable hypotheses on neural recordings and humans'/animals' behaviors in decision-making. 
For instance, given the number of input spike trains and the number of true winners, 
our results can provide an 
estimate of the minimum decision time needed, which can provide some insights on the efficiency of a WTA circuit in terms of decision time. 
As another example, when two animals are involved in the same experiment, if both animals reach the same accuracy in discriminating two objects, does the animal that decides faster have more heterogeneous distributions of input spiking activities, i.e., smaller $T_{\calR}$? 
Our results provide partial answers to this question. 


\section{Computational Model: Spiking Neuron Networks}

In this section, we provide a general description of the computation model used. There is much freedom in choosing the detailed specification of the model. In particular, in Section \ref{sec: optimal SNNs} we provide a circuit construction (for solving the $k$--WTA competition) under this computation model.

 \subsection{Network Structure}
\label{sec: SNN}
 

%
A {\em spiking neuron network} (SNN) $\calN = \pth{U, E}$ consists of a collection of neurons $U$ that are connected through synapses $E$.   
	We assume that a SNN can be conceptually partitioned into three non-overlapping layers: {\em input layer} $N_{in}$,  {\em hidden layer} $N_{h}$, and {\em output layer} $N_{out}$; the neurons in each of these layers are referred to as {\em input neurons},  {\em hidden neurons}, and {\em output neurons}, respectively. 
%
The synapses $E$ are essentially {\em directed} edges, i.e, $E : = \sth{(\nu,\nu^{\prime}): ~ \nu, \nu^{\prime}\in U}$.
For each $\nu\in U$, define $\pre_{\nu} := \sth{\nu^{\prime}: (\nu^{\prime},\nu)\in E}$ and $\pos_{\nu} := \sth{\nu^{\prime}: (\nu,\nu^{\prime})\in E}$. Intuitively, $\pre_\nu$ is the collection of neurons that can directly influence neuron $\nu$; similarly,  $\pos_\nu$ is the collection of neurons that can be directly influenced by neuron $\nu$. 
\footnote{In the languages of computational neuroscience, the incoming neighbors and outgoing neighbors are often referred to as pre-synaptic units and post-synaptic units. } We assume that the input neurons cannot be influenced by other neurons in the network, i.e., $\pre_\nu=\varnothing$ for all $\nu\in N_{in}$. 
%
%
Each edge $(\nu,\nu^{\prime})$ in $E$ has a {\em weight}, denoted by $\weight(\nu,\nu^{\prime})$.  The strength of the interaction between neuron $\nu$ and neuron $\nu^{\prime}$ is captured as $\abth{\weight(\nu,\nu^{\prime})}$. 
The sign of $\weight(\nu,\nu^{\prime})$ indicates whether neuron $\nu$ excites or inhibits neuron $\nu^{\prime}$: 
In particular, if neuron $\nu$ excites neuron $\nu^{\prime}$, then $\weight(\nu,\nu^{\prime})>0$; if 
neuron $\nu$ inhibits neuron $\nu^{\prime}$, then $\weight(\nu,\nu^{\prime})<0$. The set $E$ might contain self-loops  
with $\weight(\nu,\nu)$ capturing the self-excitatory/self-inhibitory effects.
An example of SNNs can be found in Figure \ref{fig: three layers}. 

\begin{figure}
\centering
\includegraphics[width=10cm]{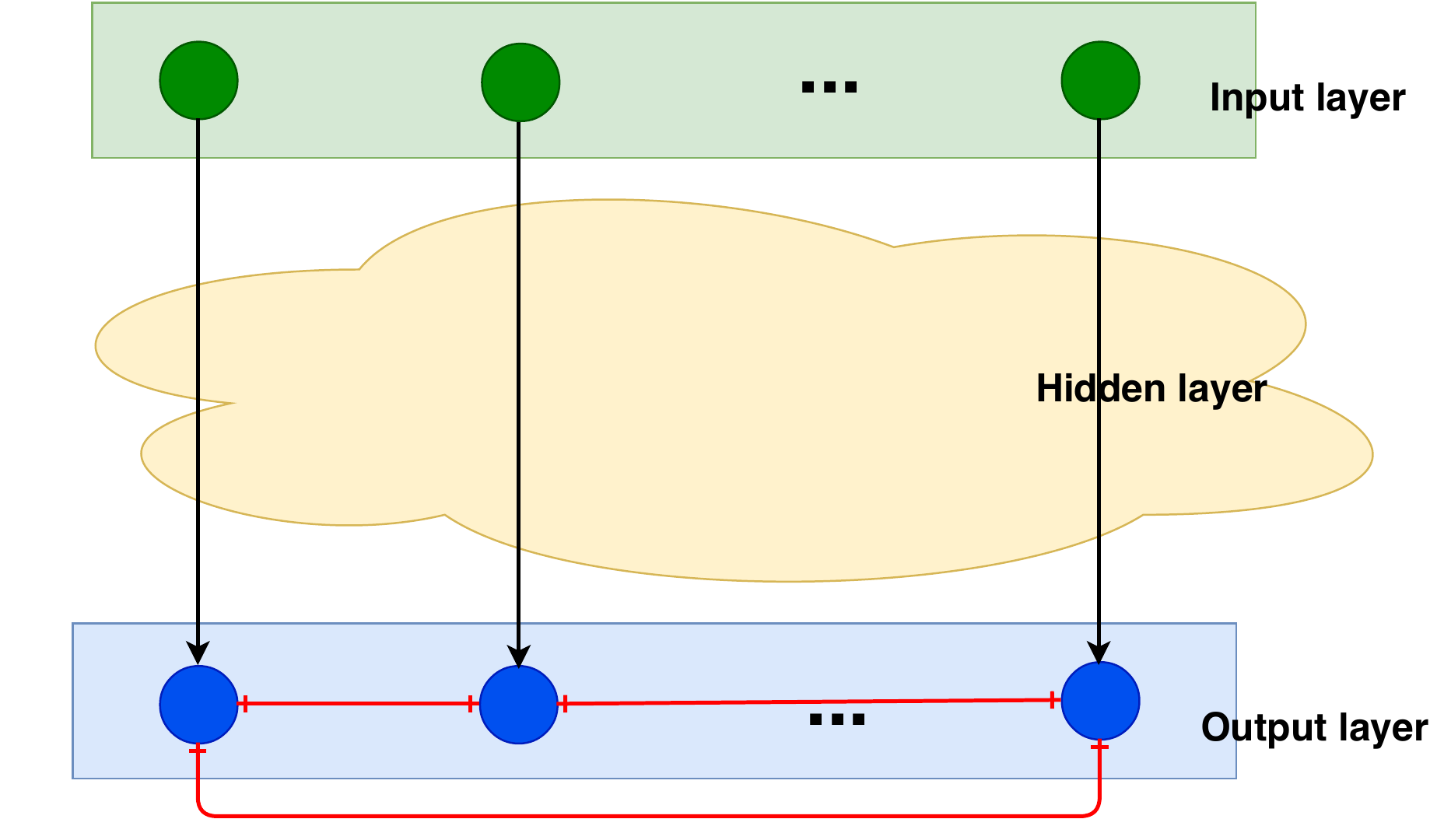}
  \caption{A SNN consists of three layers: the input layer, the output layer, and the hidden layer. The hidden neurons might connect to both the input neurons and the output neurons to assist the computation of the neuron network. Neurons are connected through synapses. 
  WTA circuits is a family of SNNs in which the number of output neurons equals the number of the input neurons. 
  }
\label{fig: three layers}
\end{figure}

\paragraph{Generic network structure for WTA circuits}
The family of WTA circuits under consideration is rather generic. We only assume that $\abth{N_{in}} = \abth{N_{out}} =n$ the numbers of the input neurons and of the output neurons are equal. 
%
For ease of exposition, denote 
\begin{align*}
N_{in} = \sth{u_1, \cdots, u_n}, ~\text{and}  ~ N_{out} = \sth{v_1, \cdots, v_n}. 
\end{align*} 
The hidden neuron subset $N_{h}$ can be arbitrary. 
%
%
The output neurons and the hidden neurons may be connected to each other in an arbitrary manner. 

\subsection{Network State}
In a SNN, the communication among neurons is abstracted as spikes. 
We assume each neuron $\nu$ has two local variables:  {\em spiking state} variable $S(\nu)$ and {\em memory state} variable $M(\nu)$.
Nevertheless, for input neurons, we only consider their spiking states, assuming that their memory states are not influenced by the dynamics of the spiking neuron network under consideration. 
We slot the time evenly with each time slot of length $1 \, ms$. Let $t=1, 2, \cdots $ be the indices of the time slots. Henceforth, by saying time $t$, we mean the time interval $[t-1, t) \, ms$.    
%
For $t\ge 1$, let $S_t(\nu) \in \sth{0, 1}$ be the spiking state of neuron $\nu$ at time $t$ indicating whether neuron $\nu$ spikes at time $t$ or not.  By convention, $S_0(\nu):=0$. 
For a non-input neuron $\nu$ and 
for $t\ge 1$, 
let $M_t(\nu)$ be the memory state of neuron $\nu$ at time $t$ summarizing the cumulative influence caused by the spikes of the neurons in $\pre_i$ during the most recent $m$ times, i.e., times $t-1, t-2, \cdots, t-m$. 
Concretely, let $V_{t}(\nu)$ be the charge of (non-input) neuron $\nu$ at time $t$ (for $t\ge 1$) defined as  
\begin{align*}
    V_{t}(\nu) : = \sum_{\nu^{\prime}\in \pre_\nu} w(\nu^{\prime},\nu) S_{t}(\nu^{\prime}). 
\end{align*}
Clearly, $V_{0}(\nu)=0$. Let $\bm{V}_{t}^\nu$ be the sequence of length $m$ such that 
\begin{align*}
    \bm{V}_{t}^\nu := \qth{V_{t}(\nu), \cdots, V_{t-m+1}(\nu)},
\end{align*}
and let $\bm{S}_{t}(\nu)$ be the sequence of length $m$ such that
\begin{align*}
\bm{S}_{t}^\nu := \qth{S_{t}(\nu), \cdots, S_{t-m+1}(\nu)}. 
\end{align*}
By convention, when $0\le t \le m$, let 
\begin{align*}
\bm{V}_{t}^\nu := \qth{V_{t}(\nu), \cdots, V_{0}(\nu), 0, \cdots, 0}
\end{align*} 
and 
\begin{align*}
\bm{S}_{t}^\nu := \qth{S_{t}(\nu), \cdots, S_{0}(\nu), 0, \cdots, 0}.  
\end{align*}
%
For $t\ge 1$, define the memory variable $M_{t}(\nu)$ as a pair of vectors $\bm{S}_{t-1}^\nu$ and $\bm{V}_{t-1}^\nu$, i.e., 
\[
M_{t}(\nu) := \pth{\bm{S}_{t-1}^\nu, \bm{V}_{t-1}^\nu}.
\]
By convention, let $M_0(\nu) :=\pth{\bm{0}, \bm{0}}$, where $\bm{0}$ is the length $m$ zero vector. 

At time $t+1$, the memory variable $M_{t+1}(\nu)$ is updated by shifting the two sequences 
forwards by one time unit -- fetching in $S_{t}(\nu)$ and $V_{t}(\nu)$, respectively, and removing $S_{t-m}(\nu)$ and $V_{t-m}(\nu)$, respectively. 
%
%
%
%
The memory state $M_t(\nu)$ is known to neuron $\nu$ only, and it can influence the probability of generating a spike at time $t$ through an activation function $\phi_\nu$, i.e.,
\begin{align}
    \label{eq: activation}
S_{t}(\nu) = \phi_\nu\pth{M_{t}(\nu)}, \forall ~ t\ge 0.  
\end{align}
Notably, $\phi_\nu$ might be a random function. 
%
%

In most neurons, the synaptic plasticity time window is about 80 -120 msec, but could also vary across brain regions, and vary across different time scales under different behavioral contexts. In a sense, the synaptic plasticity time window is closely related to $m$. 
As can be seen in Section \ref{sec: optimal SNNs}, our order-optimal WTA circuit construction requires $m$ 
to be sufficiently high. Nevertheless, this does not exclude the application of our WTA circuit to the contexts where $m$ is small. This is because the memory variable can be implemented by a chain of hidden neurons near neuron $\nu$. 
The detailed implementation of the local memory does not affect the order optimality of our WTA circuit. 

\section{Minimax Decision Accuracy/Success Probability}
\label{sec: WTA problem}

\subsection{Random Input Spike Trains} 
\label{sec: input stochasticity}
%

We study the $k$--WTA model, wherein $n$ randomly generated input spike trains are competing with each other, and, as a result of this competition, $k$ out of them are selected to be the winners. In contrast, most existing works \cite{verzi2018computing,maass1997networks,lynch2016computational} assume deterministic input spike trains. 

Recall that time is slotted into intervals of length $1\, ms$. 
%
We assume that the $n$ input spike trains are generated from 
$n$ independent Bernoulli processes with unknown parameters $p_1, \cdots, p_n$, respectively. 
We refer to $\bm{p} = \qth{p_1, \cdots, p_n}$ as a {\em rate assignment} of the WTA competition. 
For example, suppose there are 2 input spike trains with rates $0.6$ and $0.8$, respectively, i.e., $n=2$ and $\bm{p}=\qth{0.6, 0.8}$. 
In each time, with probability 0.6 the first input spike train has a spike independently from whether the second input spike train has a spike or not; similarly for the second input spike train.

Note that in the most general scenario the spikes of the input neurons might be correlated; see Section \ref{sec: discussion} for detailed comments. 
We would like to explore the more general input spikes in our future work. 



\subsection{Minimax Performance Metric} 
\label{sec: minimax WTA}

We adopt the minimax framework \cite{YHWu2017lecture} (in which the circuit designer and nature play games against each other) 
to evaluate the performance (decision accuracy versus decision time) of a WTA circuit. 

Let $\calR \subseteq [c, C]$ be an arbitrary but finite set of rates where $c$ and $C$ are two absolute constants such that $0<c<C<1$. 
A rate assignment $\bm{p}$ is chosen by nature from $\calR^n$ for which there exists a subset of $[n] := \sth{1, \cdots, n}$, denoted by $\calW (\bm{p})$, such that 
\begin{align}
\label{eq: true winners}
\abth{\calW(\bm{p})} = k, ~~ \text{and} ~ p_i > p_j ~~ \forall \, i \in \calW(\bm{p}), j\notin \calW(\bm{p})
\end{align}
-- recall that $\abth{\cdot}$ is the cardinality of a set. 
We refer to set $\calW(\bm{p})$ as the true winners with respect to the rate assignment $\bm{p}$. For example, suppose $n=5$, $k=2$, and 
\[
\bm{p} ~ = ~ \qth{p_1=0.2, ~ p_2=0.1, ~ p_3=0.2,~ p_4=0.8,~ p_5=0.85}. 
\]
Here the true winners are $4$ and $5$, i.e., $\calW(\bm{p}) = \sth{4,5}$.   
%
In this paper, we consider the following collection of rate assignments, denoted by $\calA\calR$:  
\begin{align}
\label{def: admissible rates}
\calA\calR ~ := ~ \sth{\bm{p}: \bm{p}\in \calR^n, ~ \&~ \exists \calW(\bm{p})\subseteq [n] ~ s.t. ~\abth{\calW(\bm{p})} = k, \text{and} ~ p_i > p_j ~ \forall \, i \in \calW(\bm{p}), j\notin \calW(\bm{p})}. 
\end{align}
For each of reference, we refer to an element in $\calA\calR$ as an admissible rate assignment. 
%
%
Recall that the input of a WTA circuit is a collection of $n$ independent spike trains. For a given rate assignment $\bm{p}$, let $\sth{S_t(u_i)}_{t=1}^{T}$ denote the spike train of length $T$ at input neuron $u_i$.  
%
The circuit designer wants to design a WTA circuit that outputs 
a good guess/estimate $\hat{\bm{win}}$ of $\calW(\bm{p})$ for any choice of rate assignment $\bm{p}$ in $\calA\calR$. 
Note that conditioning on 
\[
\bm{S} := \qth{\sth{S_t(u_1)}_{t=1}^{T}, \cdots, \sth{S_t(u_n)}_{t=1}^{T}},
\]
the estimate $\hat{\bm{win}}$ is independent of 
$\bm{p}$. 
%
%
Here $\bm{S}$ is used with a little abuse of notation as this notation hides its connection with $T$ and the rate parameter $\bm{p}$. \footnote{A more rigorous notation should be $\bm{S}(T, \bm{p}) := \qth{\sth{S_t(u_1)}_{t=1}^{T}, \cdots, \sth{S_t(u_n)}_{t=1}^{T}}$. We use $\bm{S}$ for $\bm{S}(T, \bm{p})$ for ease of exposition.}
Later, we use the same notation to denote the $n$ spike trains with random rate assignment, i.e., where $\bm{p}$ is randomly generated. Nevertheless, this abuse of notation significantly simplifies the exposition without sacrificing clarity. In particular, we will specify the underlying rate assignment when it is not clear from the context. 

Under minimax framework, 
we are interested in the minimax error probability \footnote{In the following expression, the $\min $ should really be an $\inf$, but we abuse notation here for ease of exposition. In addition, the $\max$ really is a $\max$ as the set $\calR$ under consideration is of finite size. }
\begin{align}
\label{eq: minimax}
\min_{\hat{\bm{win}}} \max_{\bm{p}\in \calA\calR} \prob{\hat{\bm{win}}\pth{\bm{S}} \not= \calW(\bm{p})}. 
\end{align}
For a given deterministic WTA circuit $\hat{\bm{win}}$ (i.e., the activation functions used are deterministic), the probability in $ \prob{\hat{\bm{win}}\pth{\bm{S}} \not= \calW(\bm{p})}$  is taken w.r.t.\ the randomness in the stochastic spikes of each input neuron; for a randomized 
WTA circuit $\hat{\bm{win}}$ (i.e., the activation functions are stochastic), in addition to the aforementioned source of randomness, the probability in $ \prob{\hat{\bm{win}}\pth{\bm{S}} \not= \calW(\bm{p})}$  is also taken w.r.t.\ the randomness in the activation functions.  
In \eqref{eq: minimax}, the performance metric of a WTA circuit is the worst-case error probability
\[
\max_{\bm{p}\in \calA\calR} \prob{\hat{\bm{win}}\pth{\bm{S}} \not= \calW(\bm{p})}. 
\]
Essentially, the statistical inference problem can be viewed as a game between the circuit designer and nature. 

\section{Information-Theoretic Lower Bound on Decision Time} 
\label{sec: it lower bound}
In this section, we provide a lower bound on the decision time for a given decision accuracy. The lower bounds derived in this section hold universally for all possible network structures (including the hidden layer), synapse weights, and the activation functions. 

One observation is that the decision time is naturally lower bounded by the sample complexity, 
which is closely related to the Kullback-Leibler (KL) divergence\footnote{The Kullback-Leibler (KL) divergence gauges the {\bf dissimilarity} between two distributions.} between two Bernoulli distributions. 
The KL divergence between Bernoulli random variables with parameters $r$ and $r^{\prime}$, respectively, 
is defined as 
\begin{align}
\label{eq: kl bernoulli}
d(r \parallel r^{\prime}) ~ := ~ r \log \pth{\frac{r}{r^{\prime}}} + (1-r) \log \pth{\frac{1-r}{1-r^{\prime}}},
\end{align} 
where, by convention, $ \log \frac{0}{0} :=0$. Notably, $d(\cdot \parallel \cdot)$ is not symmetric in $r$ and $r^{\prime}$. In addition, when $r\not=0$ and $r^{\prime} =0$ or $1$, $d(r \parallel r^{\prime}) = \infty$.  
Recall that set $\calR$ is an arbitrary but finite set that are contained in the interval $[c, C]$, where $c, C\in (0,1)$. It holds that $d(r \parallel r^{\prime}) < \infty$ for all $r, r^{\prime}\in \calR$. 
For the more general distributions over a common discrete alphabet $\calA$, say distributions $P$ and $Q$, the Kullback-Leibler (KL) divergence between $P$ and $Q$ is defined as follows. 
\begin{definition}[KL-divergence]
Let $\calA$ be a  discrete alphabet (finite or countably infinite), and $P$ and $Q$ be two distributions over $\calA$. Then define 
\begin{align*}
D(P\parallel Q) & : = \sum_{a\in \calA} P(a) \log \pth{\frac{P(a)}{Q(a)}},
\end{align*}
where $0\cdot \log \pth{\frac{0}{0}} =0$ by convention. 
\end{definition}
Note that $D(P\parallel Q) \ge 0$ and $D(P\parallel Q)=0$ if and only if $P=Q$ almost surely. Similar to $d(\cdot \parallel \cdot)$, $D(P\parallel Q)$ is not symmetric in $P$ and $Q$. 
In this paper, we choose the base to be 2. \footnote{Note that any base would work, see \cite[Chapter 1.1]{polyanskiy2014lecture}.} Recall that the set of admissible rate assignments $\calA\calR$ is defined in \eqref{def: admissible rates}. 

\begin{lemma}
\label{lm: m Bernoulli kl}
Fix a finite set $\calR$. Let $\bm{p} = \qth{p_1, \cdots, p_n}$ and $\bm{q} = \qth{q_1, \cdots, q_n}$ be two rate assignments in $\calA\calR$. Let $P_{\bm{S}}$ and $Q_{\bm{S}}$ be the distributions of the $n$ spike sequences of the input neurons under rate assignments $\bm{p}$ and $\bm{q}$, respectively. Then, 
\[
D(P_{\bm{S}}\parallel Q_{\bm{S}}) = T \sum_{i=1}^n d(p_i\parallel q_i). 
\]

\end{lemma}
Lemma \ref{lm: m Bernoulli kl} is proved in Appendix \ref{app: proof of kl lemma}.

For the given $\calR$, define task complexity $T_{\calR}$ 
as 
\begin{align}
\label{eq: set R}
T_{\calR} := \max_{r_1, r_2 \in \calR  \,s.t.\ r_1\not=r_2}\frac{1}{d(r_2 \parallel r_1) + d(r_1 \parallel r_2)}. 
\end{align}
It is closely related to the smallest KL divergence between two distinct statistics in $\calR$. The task complexity  $T_{\calR}$ kicks in due to the adoption of minimax decision framework \eqref{eq: minimax}.  
%
%

The following lemma is used in the proof of our information-theoretic lower bound. This is a technical supporting lemma, and the choice of the specific rate assignments 
is due to some technical convenience in proving Theorem \ref{thm: observation lower bound k_WTA}. 
\begin{lemma}
\label{lm: mutual info. k-WTA}
For any finite set $\calR$, let $r_1, r_2 \in \calR$ such that $r_1\not=r_2$. 
Let $\bm{p}^0 = \qth{p_1^0, \cdots, p_n^0}$ be 
\begin{align}
\label{def: rate 0}
p_{\ell}^0 =
\begin{cases}
r_1, ~ ~ & ~~ \text{if } \ell =1, \cdots, k;\\
r_2, & ~~ \text{otherwise}.  
\end{cases}
\end{align}
For $i=1, \cdots, k$ and $j=k+1, \cdots, n$, define rate assignment $\bm{p}^{ij}$ 
as 
\begin{align*}
p_{\ell}^{ij} =
\begin{cases}
p_{\ell}^{0}, ~ ~ & ~~ \text{if } \ell \not=i, \not=j; \\
p_{j}^{0},  ~ ~ & ~~ \text{if } \ell =i; \\
p_{i}^{0},  ~ ~ & ~~ \text{if } \ell =j.  
\end{cases}
\end{align*}
Let $X_{\bm{p}}$ be a random rate assignment. If $X_{\bm{p}}$ is uniformly distributed over 
\[
\{\bm{p}^{0}\} \cup\sth{\bm{p}^{ij}: ~ i=1, \cdots, k, \& ~ j = k+1, \cdots, n},
\] 
then the mutual information $I(X_{\bm{p}}; \bm{S})$ satisfies the following:
\[
I(X_{\bm{p}}; \bm{S}) \le T \pth{d(r_2 \parallel r_1) + d(r_1 \parallel r_2)}.
\]
\end{lemma}
%
%
%
See Appendix \ref{app: it} for definition of $I(\cdot ~ ; ~ \cdot)$. The proof of Lemma \ref{lm: mutual info. k-WTA} can be found in Appendix \ref{app:proof of mutual info}.

It turns out that if the input spike train length $T$ is not sufficiently large (specified in Theorem \ref{thm: observation lower bound k_WTA}), 
no matter how elegant the design of a WTA circuit is (no matter which activation function we choose, how many hidden neurons we use, and how we connect the hidden neurons and output neurons), its actual decision accuracy is always lower than the target decision accuracy $1-\delta$.  
%
\begin{theorem}
\label{thm: observation lower bound k_WTA}
For any $1\le k\le n-1$ and any set $\calR$ and any $\delta \in (0,1)$, if 
\[
T\le \pth{(1-\delta) \log (k(n-k) +1)-1 }T_{\calR},
\]
then 
\begin{align*}
\min_{\hat{\bm{win}}} \max_{\bm{p}\in \calA\calR} \prob{\hat{\bm{win}}\pth{\bm{S}} \not= \calW(\bm{p})} 
& ~ \ge~  \delta,   
\end{align*}
where the min is taken over all possible WTA circuits  
with different choices of activation functions and circuit architectures. 
\end{theorem}
%
Theorem \ref{thm: observation lower bound k_WTA} says that if $T< \pth{(1-\delta) \log (k(n-k) +1)-1 }T_{\calR}$, the worst case probability error of any WTA circuit is greater than $\delta$, i.e., $\max_{\bm{p}\in \calA\calR} \prob{\hat{\bm{win}}\pth{\bm{S}} \not= \calW(\bm{p})}> \delta$. 
Theorem \ref{thm: observation lower bound k_WTA} is proved in Appendix \ref{app: proof of theorem lower bound}. 
\begin{remark}[Tightness of the lower bound in Theorem \ref{thm: observation lower bound k_WTA}]
Following our line of argument, by considering a richer family of critical rate assignments in Lemma \ref{lm: mutual info. k-WTA}, we might be able to obtain a tighter lower bound. Nevertheless, the constructed WTA circuit in Section \ref{sec: optimal SNNs} turn out to be order-optimal -- its decision time matches the lower bound in Theorem \ref{thm: observation lower bound k_WTA} up to a multiplicative constant factor.
This immediately implies that the lower bound obtained in Theorem \ref{thm: observation lower bound k_WTA} is tight up to a multiplicative constant factor. 
\end{remark}

\section{Order-Optimal WTA Circuits}
\label{sec: optimal SNNs}
In Section \ref{sec: SNN}, we provided a general description of the computation model we are interested in. In this section, we construct a specific WTA circuit under this general computation model. This WTA circuit turns out to be order-optimal in terms of decision time -- the decision  time of our WTA circuit matches the lower bound in Section \ref{sec: it lower bound} up to a multiplicative constant factor. To do that, we need to specify (1) the network structure, including the number of hidden neurons, the collection of synapses (directed communication links) between neurons, and the weights of these synapses; (2) the memorization capability of each neuron, i.e., the magnitude of $m$; and (3) $\phi_\nu$ -- the activation function used by neuron $\nu$. 

\subsection{Circuit Design}
In our designed circuit, there are four parameters  $\calR$, $m$, $b$, and $\delta$, where $\calR \subseteq [c, C]$~\footnote{Recall that $c, C ~ \in (0, 1)$ are two absolute constants, i.e., they do not change with other parameters of the WTA circuit such as $n$, $k$, and $\delta$.} is a finite set from which the $p_i$'s of the input spike trains are chosen, $m$ is the memory range and $b$ is the bias at the non-input neurons, and $(1- \delta)$ is the target decision accuracy (i.e., success probability). Here, we assume that every non-input neuron has the same bias, i.e., $b_\nu = b$ for all non-input neurons $\nu$.    
%
The four parameters $\calR$, $m$, $b$, and $\delta$ can be viewed as some prior knowledge of the WTA circuit; they might be learned through some unknown network development procedure which is outside the scope of this work. In Sections \ref{subsub: Network structure}, \ref{subsub: Activation functions}, and \ref{subsub: local memory}, we present the network structure and the activation functions adopted, and the requirement on $m$. For completeness, we specify the local memory update (in particular the vector $\bm{V}$) separately in Section \ref{subsub: Local memory update}. The dynamics of our WTA circuit is summarized in Section \ref{subsub: algorithm}. 

%
 
%


%
%
\subsubsection{Network structure:}
\label{subsub: Network structure}
We propose a WTA circuit with the following network structure: 
\begin{itemize}
	\item All output neurons are connected to each other by a complete graph. That is, $(v_i,v_j)\in E$ for all $v_i,v_j \in N_{out}$ such that $v_i\not= v_j$; 
	\item Each edge from an input neuron to an output neuron has weight $1$, i.e.,   
	$\weight(u_i,v_i) = 1$ for all $u_i\in N_{h}, v_i\in N_{out}$. 
	\item All edges among the output neurons have weights $-\frac{1}{k}$. 
	That is,  
	$\weight(v_i,v_j) = - \frac{1}{k}$ for all $v_i,v_j \in N_{out}$ such that $v_i\not= v_j$. 
	\item There are no hidden neurons, i.e., $N_{h} = \emptyset$; 
\end{itemize}

\subsubsection{Update local charge vector:}
\label{subsub: Local memory update}

With the above choice of network structure, the charge $V_{t-1}(v_i)$ at the output neuron $v_i$ at time $t-1$ is 
\begin{align*}
    V_{t-1}(v_i) ~ = ~ S_{t-1}(u_i) - \frac{1}{k}\sum_{j: 1\le j\le n, \& ~ j\not=i} S_{t-1}(v_j). 
\end{align*}
Notably, $V_{t-1}(v_i) \in [-1, 1]$ for all $t\ge 1$ and output neuron $v_i$. 
%
%
When $k=1$, the above update becomes
\begin{align*}
    V_{t-1}(v_i) ~ = ~ S_{t-1}(u_i) - \sum_{j: 1\le j\le n, \& ~ j\not=i} S_{t-1}(v_j). 
\end{align*}
which can be viewed as a spike model counterpart of the potential update under the traditional continuous rate model \cite{kriener2017fast,mao2007dynamics} with lateral inhibition. 

It is easy to see the following claims hold. For brevity, their proofs are omitted. 
\begin{claim}
\label{claim: 1}
For $t\ge 1$ and for $i=1, \cdots, n$,  
$V_{t-1}(v_i)>0$ if and only if $S_{t-1}(u_i)=1$ and $\sum_{j: 1\le j\le n, \& ~ j\not=i} S_{t-1}(v_j)\le k-1$, i.e.,  
at time $t-1$, input neuron $u_i$ spikes, and fewer than $k-1$ other output neurons spike.
\end{claim}
\begin{claim}
\label{claim: 2}
For $t\ge 1$ and for $i=1, \cdots, n$, 
$V_{t-1}(v_i)\le -1$ only if $\sum_{j: 1\le j\le n, \& ~ j\not=i} S_{t-1}(v_j)\ge k$, i.e., 
at time $t-1$, more than $k$ other output neurons spike. 
\end{claim}
Note that $\sum_{j: 1\le j\le n, \& ~ j\not=i} S_{t-1}(v_j)\ge k$ is not a sufficient condition to have $V_{t-1}(v_i)\le -1$. 
To see this, suppose $\sum_{j: 1\le j\le n, \& ~ j\not=i} S_{t-1}(v_j)= k$ and $S_{t-1}(u_i)=1$. In this case it holds that $V_{t-1}(v_i)=0$. 
\begin{claim}
\label{claim: 3}
For $t\ge 1$ and for $i=1, \cdots, n$, if $V_{t-1}(v_i)=0$, one of the following holds: \\
(1) $S_{t-1}(u_i)=1$ and $\sum_{j: 1\le j\le n, \& ~ j\not=i} S_{t-1}(v_j) = k$, i.e., 
at time $t-1$, input neuron $u_i$ spikes, and exactly $k$ other output neurons spike; \\
(2) $S_{t-1}(u_i)=0$ and $\sum_{j: 1\le j\le n, \& ~ j\not=i} S_{t-1}(v_j) = 0$, i.e.,
at time $t-1$, input neuron $u_i$ does not spike, and no other output neurons spike. 
\end{claim}


\subsubsection{Activation functions:}
\label{subsub: Activation functions}
There are many different choices of activation functions; see \cite{wikiacf} for a detailed list. In our construction, we use a simple threshold activation function, 
i.e., 
\begin{align*}
S_{t}(v_i) =  
\begin{cases}
1, & \text{if } (b-1) \indc{S_{t-1}(v_i) =1} + \qth{
\sum_{r=1}^m\indc{V_{t-r}(v_i) >0} - m \sum_{r=1}^{m} \indc{V_{t-r}(v_i) \le -1}}_+\ge b; \\
0, & \text{otherwise,} 
\end{cases} 
\end{align*}
$\qth{\cdot}_+ = \max \qth{\cdot, 0}$, and $b>0$ is the bias at neuron $v_i$ for $i=1, \cdots, n$. 
It is easy to see that this activation function falls under the general form given by \eqref{eq: activation}. 

\begin{remark}
If the output neuron $v_i$ does not spike at time $t-1$, i.e., $S_{t-1}(v_i) =0$, then in order for $v_i$ to spike at time $t$, the following needs to hold: 
\[
\qth{
\sum_{r=1}^m \indc{V_{t-r}(v_i) >0} - m \sum_{r=1}^{m} \indc{V_{t-r}(v_i) \le -1}}_+\ge b. 
\]
In contrast, if the output neuron $v_i$ does spike at time $t-1$, i.e., $S_{t-1}(v_i) =1$, then 
\[
\qth{
\sum_{r=1}^m \indc{V_{t-r}(v_i) >0} - m \sum_{r=1}^{m} \indc{V_{t-r}(v_i) \le -1}}_+\ge 1
\]
is enough for $v_i$ to spike at time $t$. That is, under our activation rule, $S_{t-1}(v_i) =1$ makes the activation of $v_i$ much easier in the next round. However, if there exists $r \in \sth{1, 2, \cdots, m}$ such that 
\[
\indc{V_{t-r}(v_i) \le -1}=1,
\]
then 
\begin{align*}
\sum_{r=1}^m \indc{V_{t-r}(v_i) >0} - m \sum_{r=1}^{m} \indc{V_{t-r}(v_i) \le -1} & \le \sum_{r=1}^m \indc{V_{t-r}(v_i) >0} - m \\
& \le m - m = 0. 
\end{align*}
Thus, 
\begin{align*}
&(b-1) \indc{S_{t-1}(v_i) =1} + \qth{\sum_{r=1}^m \indc{V_{t-r}(v_i) >0} - m \sum_{r=1}^{m} \indc{V_{t-r}(v_i) \le -1}}_+\\
& = (b-1) \indc{S_{t-1}(v_i) =1} + 0 \\
&\le b-1 < b, 
\end{align*}
i.e., the output neuron $v_i$ does not spike at time $t$. In other words, as long as there exists $r \in \sth{1, 2, \cdots, m}$ such that $
\indc{V_{t-r}(v_i) \le -1}=1$, the activation of $v_i$ is inhibited at time $t$.

\end{remark}

\subsubsection{Local memorization capability:}
\label{subsub: local memory}
%
In our proposed circuit, we require that $m$ satisfies the following: 
\begin{align}
\label{eq: high h}
m \ge \frac{8C^2(1-c)}{c^2(1-C)}\pth{\log \pth{\frac{3}{\delta}} + \log k(n-k) } T_{\calR} ~ := ~ m^* 
\end{align}
for target decision accuracy $1-\delta \in (0, 1)$. In addition, we set $b=cm^*$. 
Recall that $c, C\in (0,1)$ are two absolute constants that are lower bound and upper bound of any $\calR$, respectively. 

Intuitively, when other parameters are fixed, the higher the desired accuracy (i.e., the smaller $\delta$) , the larger $m^*$, i.e., the more memory is needed for selecting the winners in our WTA circuit. Similarly, the easier to distinguish two spike trains with different statistics (i.e., the lower $T_{\calR}$), the smaller $m^*$. Interesting, with other parameters fixed, $m^*$ depends on $k$ as follows:  $m^*$ is increasing in $k$ when $k\in \sth{1, \cdots, \lfloor \frac{n}{2}\rfloor }$, and $m^*$ is decreasing in $k$ when $k\in \sth{\lceil \frac{n}{2}\rceil, \cdots, n-1}$. In many practical settings we care about the region where $k\ll n$.     
Besides, with the choice of bias $b=cm^*$, the larger $m^*$ also implies longer time is needed for our WTA circuit to declare $k$ winners; details can be found (1) in Theorem \ref{thm: k wta threhold real}.

On the other hand, in most neurons the synaptic plasticity time window is about 80-120 ms, and it is unclear whether \eqref{eq: high h} can be immediately satisfied or not. Fortunately, even if \eqref{eq: high h} is not immediately satisfied by a neuron due to its local bio-plausibility, it is possible that its local memory might be realized 
using a chain of hidden neurons. 

\subsubsection{Algorithm \ref{alg: k WTA}}
\label{subsub: algorithm}
The dynamics of our WTA circut is summarized in Algorithm \ref{alg: k WTA}, which is fully determined by what has been described in Sections \ref{subsub: Network structure}, \ref{subsub: Local memory update}, \ref{subsub: Activation functions}, and \ref{subsub: local memory}. 
For Algorithm \ref{alg: k WTA}, we declare the first $k$ output neurons that spike simultaneously to be winners. 
\begin{algorithm}
\caption{$k$--WTA} 
\label{alg: k WTA}
\LinesNumbered
{\em \bf Input:} $\calR$, $m$, $b$, and $\delta$. 


\vskip 0.6\baselineskip 

\For{$t\ge 1$}{
{\bf  At output neuron $v_i$ for $i=1, \cdots, n$:}
%
$V_{t-1}(v_i) \gets S_{t-1}(u_i) - \frac{1}{k}\sum_{j: 1\le j \le n, \& j\not=i} S_{t-1}(v_j)$\;  

~ $\bm{V}_{t-1}(v_i) \gets \qth{V_{t-1}(v_i), V_{t-2}(v_i), \cdots, V_{t-m}(v_i)}$\; 

~ $\bm{S}_{t-1}(v_i) \gets \qth{S_{t-1}(v_i), S_{t-2}(v_i), \cdots, S_{t-m}(v_i)}$\;  

~ $M_t(v_i) \gets \pth{\bm{V}_{t-1}(v_i), \bm{S}_{t-1}(v_i)}$\;   
 \eIf{$(b-1)\indc{S_{t-1}(v_i) =1} + \qth{\sum_{r=1}^m \indc{V_{t-r}(v_i) >0} - m \sum_{r=1}^{m} \indc{V_{t-r}(v_i) \le -1}}_+ \ge b$}
 {$S_{t}(v_i)\gets 1$.}
 {$S_{t}(v_i)\gets 0$.}
}
\vskip 0.2\baselineskip 

%
\end{algorithm}









%
\subsection{Circuit Performance}
\label{subsec: main results}
Recall that $\calW(\bm{p})$ and $m^*$ are defined in \eqref{eq: true winners} and \eqref{eq: high h}, respectively. 
\begin{theorem}
\label{thm: k wta threhold real}
Fix $\delta \in (0, 1]$, and $1\le k \le n-1$. Choose $m\ge m^*$ and $b= \max\sth{cm^*, 2}$. 
Then for any admissible rate assignment $\bm{p}$, with probability at least $1-\delta$, the following hold: 
\begin{itemize}
\item[(1)] There exist $k$ output neurons that spike simultaneously 
by time $m^*$.  
\item[(2)] The first set of such $k$ output neurons are the true winners $\calW(\bm{p})$.  
\item[(3)] From the first time in which these $k$ output neurons spike simultaneously, these $k$ output neurons spike consecutively for at least $b$ times, and no other output neurons can spike within $b$ times. 
\end{itemize}
\end{theorem}
The proof of Theorem \ref{thm: k wta threhold real} can be found in Appendix \ref{app: proof of algorithm}. 
The first bullet in Theorem \ref{thm: k wta threhold real} implies that our WTA circuit can provide an output (a selection of $k$ output neurons) by time $m^*$; the second bullet in Theorem \ref{thm: k wta threhold real} says that the circuit's output indeed corresponds to the $k$ true winners; and the third bullet says that the $k$ simultaneous spikes of the selected winners are stable -- the $k$ selected winners continue to spike consecutively for at least $b$ times. 
%
The proof of Theorem \ref{thm: k wta threhold real} essentially says that with high probability, under Algorithm \ref{alg: k WTA}, the number of output neurons that spike simultaneously is monotonically increasing until it reaches $k$. Upon the simultaneous spike of $k$ output neurons, by our threshold activation rule, we know that the other output neurons are likely to be inhibited. In particular, if these $k$ output neurons are the first $k$ output neurons that spike simultaneously, then the activation of the other output neurons are likely to be inhibited for at least $b$ times. 
%
\begin{remark}[Controlling stability]
As can be seen from the proof of Theorem \ref{thm: k wta threhold real}, in the activation function of Algorithm \ref{alg: k WTA}  
\[
(b-1)\indc{S_{t-1}(v_i) =1} + \qth{\sum_{r=1}^m \indc{V_{t-r}(v_i) >0} - m \sum_{r=1}^{m} \indc{V_{t-r}(v_i) \le -1}}_+ \ge b
\]
the first term $(b-1)\indc{S_{t-r}(v_i) =1}$ is crucial in achieving (3) in Theorem \ref{thm: k wta threhold real}. In fact, we can increase the stability period by introducing a stability parameter $s$ such that $1<s\le m$ and modifying the activation rule. Details can be found in Algorithm \ref{alg: k WTA modified}. It is easy to see that the activation function falls under the general form in \eqref{eq: activation}. 
In the new activation function in Algorithm \ref{alg: k WTA modified}, for output neuron $v_i$, once it spikes, it continues to spike for at least $s$ times.  
Following our line of analysis in the proof of Theorem \ref{thm: k wta threhold real}, it can be seen that the declared $k$ winners, from the first time they spike simultaneously, continue to spike consecutively for at least $s$ times.  
\begin{algorithm}
\caption{$k$--WTA} 
\label{alg: k WTA modified}
{\em \bf Input:} $\calR$, $m$, $b$, $\delta$, and $s$ where $1< s\le m$.  

\vskip 0.6\baselineskip 

\For{$t\ge 1$}
{
\vskip 0.3\baselineskip 

{\bf  At output neuron $v_i$ for $i=1, \cdots, n$:} 

~ $V_{t-1}(v_i) \gets S_{t-1}(u_i) - \frac{1}{k}\sum_{j: 1\le j \le n, \& j\not=i} S_{t-1}(v_j)$\;  
~ $\bm{V}_{t-1}(v_i) \gets \qth{V_{t-1}(v_i), V_{t-2}(v_i), \cdots, V_{t-m}(v_i)}$\; 

~ $\bm{S}_{t-1}(v_i) \gets \qth{S_{t-1}(v_i), S_{t-2}(v_i), \cdots, S_{t-m}(v_i)}$ \; 

~ $M_t(v_i) \gets \pth{\bm{V}_{t-1}(v_i), \bm{S}_{t-1}(v_i)}$.

 \eIf{$\qth{\sum_{r=1}^m \indc{V_{t-r}(v_i) >0} - m \sum_{r=1}^{m} \indc{V_{t-r}(v_i) \le -1}}_+ \ge b$}
 {$S_{t}(v_i)\gets 1$.}
 {\eIf{$S_{t-1}(v_i)=1$ and $\exists \, r\in \sth{2, \cdots, s}$ such that $S_{t-r}(v_i)=0$}
 {$S_{t}(v_i)\gets 1$.}
 {$S_{t}(v_i)\gets 0$.}
}
}

\vskip 0.2\baselineskip 

\end{algorithm}

\end{remark}

\begin{remark}[Order-optimality]
\label{rmk: order-optimal WTA}
The decision time performance stated in (1) of Theorem \ref{thm: k wta threhold real} matches the information-theoretical lower bound in Theorem \ref{thm: observation lower bound k_WTA} up to a  multiplicative constant factor both (a) when $\delta$ is sufficiently small and does not depend on $n$, $k$, $T_{\calR}$, $c$, and $C$, and (b) when $\delta$ decays to zero at a speed at most $\frac{1}{(k(n-k))^{c_0}}$ where $c_0>0$ is some fixed constant. 
The detailed order-optimality argument is given next. 

\paragraph{Suppose that $\delta$ is sufficiently small and does not depend on $n$, $k$, $T_{\calR}$, $c$, and $C$}
Here, for ease of exposition, we illustrate the order-optimality with a specific choice of $\delta$. In fact, the order-optimality holds generally for constant $\delta\in (0,1)$ as long as it does not depend on $n$, $k$, $T_{\calR}$, $c$, and $C$. 

Suppose the target decision accuracy is $1-\delta = 0.9$, i.e., $\delta = 0.1$. Then as long as $n\ge 31$, for any $1\le k \le n-1$, %
\begin{align*}
m^* = \frac{8C^2(1-c)}{c^2(1-C)} \pth{\log \frac{3}{0.1}  + \log k(n-k)}T_{\calR}  \le \frac{16C^2(1-c)}{c^2(1-C)} \log k(n-k) T_{\calR}. 
\end{align*}
On the other hand, recall from Theorem \ref{thm: observation lower bound k_WTA} that to have $\delta = 0.1$, the decision time is no less than 
\begin{align*}
\pth{(1-\delta)\log(k(n-k)+1)-1} T_{\calR}  \ge \frac{1}{2} \log(k(n-k)+1)T_{\calR} \ge \frac{1}{2} \log k(n-k)T_{\calR}
\end{align*}
where the first inequality holds as long as $n\ge 8$. Thus, when $n\ge 31$, in order to achieve the decision accuracy $1-\delta = 0.9$, the decision time of 
our WTA circuit is on the same order of the information-theoretic lower bound in Theorem \ref{thm: observation lower bound k_WTA}.

\paragraph{Suppose $\delta$ decays to zero at a moderate speed}
The decision time of our WTA circuit is order-optimal even for diminishing decision error $\delta$ as long as $\delta = \Omega (\frac{3}{(k(n-k))^{c_0}})$ where $c_0>0$ -- it does not decay to zero ``too fast'' in $k(n-k)$. To see this, let $\delta = \frac{3}{\pth{k(n-k)}^{c_0}}$ for some constant $c_0>0$. We have 
\begin{align}
\frac{8C^2(1-c)}{c^2(1-C)}\pth{\log \pth{\frac{3}{\frac{3}{(k(n-k))^{c_0}}}} + \log k(n-k)} T_{\calR} = \frac{8C^2(1-c)(c_0+1)}{c^2(1-C)} \log k(n-k) T_{\calR}. 
\end{align}
\end{remark}

\paragraph{Resetting circuit when the input spike trains become quiescent}
In Algorithm \ref{alg: k WTA}, if the input spike trains become quiescent, then the corresponding circuits  also become quiescent despite some delay in this response. 
\begin{lemma}
\label{lm: quiescent}
If all input neurons are quiescent at time $t_0$, and remain to be quiescent for all $t\ge t_0$,  
then $V_{t}(v_i)=0$ and $S_t(v_i)=0$ for any $t> t_0+m$.   
\end{lemma}
Lemma \ref{lm: quiescent} is proved in Appendix \ref{app: lm: quiescent}.

\section{Discussion}
\label{sec: discussion}
In this paper, we investigated how $k$-WTA computation is robustly achieved in the presence of inherent noise in the input spike trains. In a spike-based $k$-WTA model, $n$ randomly generated input spike trains are competing with each other, and the top $k$ neurons with highest underlying statistics are the true winners. Given the stochastic nature of the spike trains, it is not trivial to properly select winners among a group of neurons. We 
derived an information-theoretic lower bound on the decision time for a given decision accuracy. Notably, this lower bound holds universally for any WTA circuit that falls within our model framework, regardless of their circuit architectures or their adopted activation functions. Furthermore, we constructed a circuit whose decision time matches this lower bound up to a constant multiplicative factor, suggesting that our derived lower bound is order-optimal. Here the order-optimality is stated in terms of its scaling in $n$, $k$, and $T_{\calR}$. 

\subsection{Comparison to previous WTA models}

Randomness is introduced at different stages of brain computation and the stochastic nature of the spike trains are well observed \cite{baddeley1997responses, kara2000low, maimon2009beyond, ferrari2018simple}. In our work, we focused on how to robustly achieve $k$-WTA computation in face of the intrinsic randomness in the spike trains. A common WTA model assumes that neurons transmit information by a continuous variable such as firing rate \cite{dayan2001theoretical}, which ignores the intrinsic randomness in spiking trains. Although some studies used additive Gaussian noise \cite{kriener2017fast, li2013class, lee1999attention, rougier2006emergence} in their rate-based WTA circuits to account for input randomness, these circuits are usually very sensitive to noise and could not successfully select even a single winner unless additional non-linearity is added \cite{kriener2017fast}. 
In fact, a neuron with a second non-linearity is similar to an output neuron in our constructed WTA circuit in that they both integrate their local inputs. Unfortunately, only simulation results were provided in \cite{kriener2017fast}; a theoretical justification of why such second non-linearity makes their WTA circuit robust to input noise is lacking. 
Though we focused on spike-based model, we hope our results can provide some insights for the rate-based model as well. 
On top of that, a rate-based model would require a high communication bandwidth, yet communication bandwidth is limited in the brain. Our spiking neural network model captures this feature by having a low communication cost, since it broadcasts 1 bit only. 

However, we did not try to model every biologically relevant feature. In several studies using spiking network models, individual units are often modeled with details like ion channels and specific synaptic connectivity. Though more biologically relevant than our spiking neuron network model, those details significantly complicate the analysis. In fact, it could be challenging and intricate to move beyond computer simulation to characterize the model dynamics (such as the spiking nature of each unit, the time it takes to stabilize, etc.)
analytically.

\subsection{Potential applications for physiological experiments}

Our work further provided testable hypotheses on how network size, similarities between input spike trains, and synaptic memory capacity would affect this lower bound. For example, in behavioral experiments using electrolytic lesions or pharmacological inhibition \cite{clark2003contributions, hanks2006microstimulation, yttri2013lesions, katz2016dissociated}, the changes in performance are often highly variable and nonlinear. One possible reason comes from the difficulty of precisely manipulating network size as well as a lack of theoretical description of the relationship between network size and performance. With our analytical characterization, one might be able to estimate changes in the effective network size given performance in a decision-making task. 

Besides the effect of network size, the distribution of feature representations (i.e., different set $\calR$s of different individual animals) could be used to account for between-subject variability in decision making. 
Consider a random-dot coherent motion task where animals need to decide which of two directions the majority of dots are moving \cite{shadlen2001neural}. In this task, performance accuracy and reaction time vary across animals. If we perform neural recordings in their visual cortex (i.e., to record their $\calR$s), we might be able to decode their reaction time or accuracy, given population representations of dot motion in these cortical neurons \cite{shadlen1996motion, jazayeri2006optimal}. For example, an animal whose stimulus-evoked responses are more heterogeneous in the visual cortex might be able to react faster given the same accuracy, governed by our derived lower-bound. 

Last but not least, our work also offered predictions on how local memory capacity could affect performance in decision-making. For example, when there is more ambiguity in input representations, to obtain the same performance (both accuracy and decision time), a larger time window for memory storage in synapses \cite{knoblauch2010memory} is required. From previous experimental work \cite{bittner2017behavioral}, we know that synaptic plasticity has time scale ranging from milliseconds to seconds across different brain regions, and such plasticity could efficiently store entire behavioral sequences within synaptic weights. Combining with our analytical characterization, when performance accuracy changes over time, assuming other parameters such as input statistics, decision time and network size are fixed, one might be able to predict how synaptic plasticity changes. Overall, our work not only provided a theoretical framework, but also provided a set of testable hypotheses on neural recordings and behaviors in decision-making under ambiguity. 

\subsection{Limitations and extensions}
When $\delta$ is a constant, our lower bound is order-optimal in terms of its scaling in $n$, $k$, and $T_{\calR}$. Nevertheless, the scaling of the derived lower bound in terms of $\delta$ is not tight. 
It would be interesting to know the optimal scaling in $\delta$ when other parameters ($n$, $k$, and $T_{\calR}$) are fixed. We leave it as one future direction. 

To simplify complexity, our model posed a few assumptions that ignored some features in the brain. One of these assumptions is that each input neuron is independent. However, various degrees of average noise correlations between cortical neurons have been reported. For example, average noise correlations in primary visual cortex could be close to 0.1 \cite{scholvinck2015cortical}, 0.18 \cite{smith2008spatial}, or even much larger as 0.35 \cite{gutnisky2008adaptive}. Similarly, noise correlations have been observed in other sensory brain regions \cite{cohen2011measuring}. In our work, we ignored correlations between these neurons, but it would be interesting as a future direction to extend in our spiking network model. 

Second, our model used a threshold activation function by assuming the synaptic transmission is basically noise-free and that the only noise source comes from the input in this paper. However, synaptic transmission is highly unreliable in biological networks \cite{allen1994evaluation, faisal2008noise, borst2010low}, and a deterministic activation function would fail to capture this feature compared to a stochastic activation function.
Moreover, failure in synaptic transmission could serve a computational role \cite{branco2009probability,maass1997networks}. 

Another assumption in our circuit is that the output neurons can inhibit each other. In common scenarios, an output neuron is usually excitatory, and does not inhibit other neurons directly without recruiting inhibitory cells. We incorporate stability in these output neurons by assuming they can inhibit each other in our circuit implementation. For a model where an output neuron is limited to be excitatory only, we can add a chain of inhibitory neurons to achieve stability WTA computation. 

Last but not least, in our $k$-WTA circuit, the number of output neurons that spike simultaneously increases monotonically until there are exactly $k$ output neurons that spike simultaneously. We acknowledge that this might not be biologically plausible in most cases in the brain. From large-scale neural recordings, we know that the number of neurons that spike simultaneously is usually variable, so this could be a future direction to construct a circuit that better matches experimental observations.

\section*{Acknowledgement}
We would like to thank Christopher Quinn at Purdue University and Zhi-Hong Mao at University of Pittsburgh for the helpful discussions and references.

\bibliographystyle{alpha}
\bibliography{WTABib}

\newcommand{\etalchar}[1]{$^{#1}$}
\begin{thebibliography}{POMT{\etalchar{+}}02}

\bibitem[AS94]{allen1994evaluation}
Christina Allen and Charles~F Stevens.
\newblock An evaluation of causes for unreliability of synaptic transmission.
\newblock {\em Proceedings of the National Academy of Sciences},
  91(22):10380--10383, 1994.

\bibitem[BAB{\etalchar{+}}97]{baddeley1997responses}
Roland Baddeley, Larry~F Abbott, Michael~CA Booth, Frank Sengpiel, Tobe
  Freeman, Edward~A Wakeman, and Edmund~T Rolls.
\newblock Responses of neurons in primary and inferior temporal visual cortices
  to natural scenes.
\newblock {\em Proceedings of the Royal Society of London B: Biological
  Sciences}, 264(1389):1775--1783, 1997.

\bibitem[BC95]{buzsaki1995temporal}
Gy{\"o}rgy Buzs{\'a}ki and James~J Chrobak.
\newblock Temporal structure in spatially organized neuronal ensembles: a role
  for interneuronal networks.
\newblock {\em Current opinion in neurobiology}, 5(4):504--510, 1995.

\bibitem[BIM98]{berry1998refractoriness}
Michael~J Berry~II and Markus Meister.
\newblock Refractoriness and neural precision.
\newblock In {\em Advances in Neural Information Processing Systems}, pages
  110--116, 1998.

\bibitem[BMG{\etalchar{+}}17]{bittner2017behavioral}
Katie~C Bittner, Aaron~D Milstein, Christine Grienberger, Sandro Romani, and
  Jeffrey~C Magee.
\newblock Behavioral time scale synaptic plasticity underlies ca1 place fields.
\newblock {\em Science}, 357(6355):1033--1036, 2017.

\bibitem[Bor10]{borst2010low}
J~Gerard~G Borst.
\newblock The low synaptic release probability in vivo.
\newblock {\em Trends in neurosciences}, 33(6):259--266, 2010.

\bibitem[BS09]{branco2009probability}
Tiago Branco and Kevin Staras.
\newblock The probability of neurotransmitter release: variability and feedback
  control at single synapses.
\newblock {\em Nature Reviews Neuroscience}, 10(5):373, 2009.

\bibitem[CK11]{cohen2011measuring}
Marlene~R Cohen and Adam Kohn.
\newblock Measuring and interpreting neuronal correlations.
\newblock {\em Nature neuroscience}, 14(7):811, 2011.

\bibitem[CMA{\etalchar{+}}03]{clark2003contributions}
Luke Clark, Facundo Manes, Nagui Antoun, Barbara~J Sahakian, and Trevor~W
  Robbins.
\newblock The contributions of lesion laterality and lesion volume to
  decision-making impairment following frontal lobe damage.
\newblock {\em Neuropsychologia}, 41(11):1474--1483, 2003.

\bibitem[DA01]{dayan2001theoretical}
Peter Dayan and Laurence~F Abbott.
\newblock Theoretical neuroscience: computational and mathematical modeling of
  neural systems.
\newblock 2001.

\bibitem[EB02]{ernst2002humans}
Marc~O Ernst and Martin~S Banks.
\newblock Humans integrate visual and haptic information in a statistically
  optimal fashion.
\newblock {\em Nature}, 415(6870):429, 2002.

\bibitem[FDMM18]{ferrari2018simple}
Ulisse Ferrari, Stephane Deny, Olivier Marre, and Thierry Mora.
\newblock A simple model for low variability in neural spike trains.
\newblock {\em arXiv preprint arXiv:1801.01362}, 2018.

\bibitem[FSW08]{faisal2008noise}
A~Aldo Faisal, Luc~PJ Selen, and Daniel~M Wolpert.
\newblock Noise in the nervous system.
\newblock {\em Nature reviews neuroscience}, 9(4):292, 2008.

\bibitem[GD08]{gutnisky2008adaptive}
Diego~A Gutnisky and Valentin Dragoi.
\newblock Adaptive coding of visual information in neural populations.
\newblock {\em Nature}, 452(7184):220, 2008.

\bibitem[GKvHW96]{gerstner1996neuronal}
Wulfram Gerstner, Richard Kempter, J~Leo van Hemmen, and Hermann Wagner.
\newblock A neuronal learning rule for sub-millisecond temporal coding.
\newblock {\em Nature}, 383(6595):76, 1996.

\bibitem[HDS06]{hanks2006microstimulation}
Timothy~D Hanks, Jochen Ditterich, and Michael~N Shadlen.
\newblock Microstimulation of macaque area lip affects decision-making in a
  motion discrimination task.
\newblock {\em Nature neuroscience}, 9(5):682, 2006.

\bibitem[HDZ08]{hromadka2008sparse}
Tom{\'a}{\v{s}} Hrom{\'a}dka, Michael~R DeWeese, and Anthony~M Zador.
\newblock Sparse representation of sounds in the unanesthetized auditory
  cortex.
\newblock {\em PLoS biology}, 6(1):e16, 2008.

\bibitem[HSM06]{hahn2006phase}
Thomas~TG Hahn, Bert Sakmann, and Mayank~R Mehta.
\newblock Phase-locking of hippocampal interneurons' membrane potential to
  neocortical up-down states.
\newblock {\em Nature neuroscience}, 9(11):1359, 2006.

\bibitem[HW59]{hubel1959receptive}
David~H Hubel and Torsten~N Wiesel.
\newblock Receptive fields of single neurones in the cat's striate cortex.
\newblock {\em The Journal of physiology}, 148(3):574--591, 1959.

\bibitem[HW98]{harris1998signal}
Christopher~M Harris and Daniel~M Wolpert.
\newblock Signal-dependent noise determines motor planning.
\newblock {\em Nature}, 394(6695):780, 1998.

\bibitem[IKN98]{itti1998model}
Laurent Itti, Christof Koch, and Ernst Niebur.
\newblock A model of saliency-based visual attention for rapid scene analysis.
\newblock {\em IEEE Transactions on pattern analysis and machine intelligence},
  20(11):1254--1259, 1998.

\bibitem[JB67]{jacobs1967lower}
I~Jacobs and E~Berlekamp.
\newblock A lower bound to the distribution of computation for sequential
  decoding.
\newblock {\em IEEE Transactions on Information Theory}, 13(2):167--174, 1967.

\bibitem[JM06]{jazayeri2006optimal}
Mehrdad Jazayeri and J~Anthony Movshon.
\newblock Optimal representation of sensory information by neural populations.
\newblock {\em Nature neuroscience}, 9(5):690, 2006.

\bibitem[KCF17]{kriener2017fast}
Birgit Kriener, Rishidev Chaudhuri, and Ila Fiete.
\newblock How fast is neural winner-take-all when deciding between many
  options?
\newblock {\em bioRxiv}, page 231753, 2017.

\bibitem[KF08]{karlsson2008network}
Mattias~P Karlsson and Loren~M Frank.
\newblock Network dynamics underlying the formation of sparse, informative
  representations in the hippocampus.
\newblock {\em Journal of Neuroscience}, 28(52):14271--14281, 2008.

\bibitem[KK01]{kinoshita2001neural}
Masaharu Kinoshita and Hidehiko Komatsu.
\newblock Neural representation of the luminance and brightness of a uniform
  surface in the macaque primary visual cortex.
\newblock {\em Journal of neurophysiology}, 86(5):2559--2570, 2001.

\bibitem[KP04]{knill2004bayesian}
David~C Knill and Alexandre Pouget.
\newblock The bayesian brain: the role of uncertainty in neural coding and
  computation.
\newblock {\em TRENDS in Neurosciences}, 27(12):712--719, 2004.

\bibitem[KPS10]{knoblauch2010memory}
Andreas Knoblauch, G{\"u}nther Palm, and Friedrich~T Sommer.
\newblock Memory capacities for synaptic and structural plasticity.
\newblock {\em Neural Computation}, 22(2):289--341, 2010.

\bibitem[KRR00]{kara2000low}
Prakash Kara, Pamela Reinagel, and R~Clay Reid.
\newblock Low response variability in simultaneously recorded retinal,
  thalamic, and cortical neurons.
\newblock {\em Neuron}, 27(3):635--646, 2000.

\bibitem[KTA{\etalchar{+}}03]{kourtzi2003integration}
Zoe Kourtzi, Andreas~S Tolias, Christian~F Altmann, Mark Augath, and Nikos~K
  Logothetis.
\newblock Integration of local features into global shapes: monkey and human
  fmri studies.
\newblock {\em Neuron}, 37(2):333--346, 2003.

\bibitem[KW04]{kording2004bayesian}
Konrad~P K{\"o}rding and Daniel~M Wolpert.
\newblock Bayesian integration in sensorimotor learning.
\newblock {\em Nature}, 427(6971):244, 2004.

\bibitem[KYPH16]{katz2016dissociated}
Leor~N Katz, Jacob~L Yates, Jonathan~W Pillow, and Alexander~C Huk.
\newblock Dissociated functional significance of decision-related activity in
  the primate dorsal stream.
\newblock {\em Nature}, 535(7611):285, 2016.

\bibitem[LCG{\etalchar{+}}15]{li2015motor}
Nuo Li, Tsai-Wen Chen, Zengcai~V Guo, Charles~R Gerfen, and Karel Svoboda.
\newblock A motor cortex circuit for motor planning and movement.
\newblock {\em Nature}, 519(7541):51, 2015.

\bibitem[LIKB99]{lee1999attention}
Dale~K Lee, Laurent Itti, Christof Koch, and Jochen Braun.
\newblock Attention activates winner-take-all competition among visual filters.
\newblock {\em Nature neuroscience}, 2(4):375, 1999.

\bibitem[LLW13]{li2013class}
Shuai Li, Yangming Li, and Zheng Wang.
\newblock A class of finite-time dual neural networks for solving quadratic
  programming problems and its k-winners-take-all application.
\newblock {\em Neural Networks}, 39:27--39, 2013.

\bibitem[LMP16]{lynch2016computational}
Nancy Lynch, Cameron Musco, and Merav Parter.
\newblock Computational tradeoffs in biological neural networks:
  Self-stabilizing winner-take-all networks.
\newblock {\em arXiv preprint arXiv:1610.02084}, 2016.

\bibitem[MA09]{maimon2009beyond}
Gaby Maimon and John~A Assad.
\newblock Beyond poisson: increased spike-time regularity across primate
  parietal cortex.
\newblock {\em Neuron}, 62(3):426--440, 2009.

\bibitem[Maa97]{maass1997networks}
Wolfgang Maass.
\newblock Networks of spiking neurons: the third generation of neural network
  models.
\newblock {\em Neural networks}, 10(9):1659--1671, 1997.

\bibitem[Maa00]{maass2000computational}
Wolfgang Maass.
\newblock On the computational power of winner-take-all.
\newblock {\em Neural computation}, 12(11):2519--2535, 2000.

\bibitem[MCM07]{majaj2007motion}
Najib~J Majaj, Matteo Carandini, and J~Anthony Movshon.
\newblock Motion integration by neurons in macaque mt is local, not global.
\newblock {\em Journal of Neuroscience}, 27(2):366--370, 2007.

\bibitem[MM07]{mao2007dynamics}
Zhi-Hong Mao and Steve~G Massaquoi.
\newblock Dynamics of winner-take-all competition in recurrent neural networks
  with lateral inhibition.
\newblock {\em IEEE transactions on neural networks}, 18(1):55--69, 2007.

\bibitem[Nel04]{nelken2004processing}
Israel Nelken.
\newblock Processing of complex stimuli and natural scenes in the auditory
  cortex.
\newblock {\em Current opinion in neurobiology}, 14(4):474--480, 2004.

\bibitem[OF04]{olshausen2004sparse}
Bruno~A Olshausen and David~J Field.
\newblock Sparse coding of sensory inputs.
\newblock {\em Current opinion in neurobiology}, 14(4):481--487, 2004.

\bibitem[PG99]{platt1999neural}
Michael~L Platt and Paul~W Glimcher.
\newblock Neural correlates of decision variables in parietal cortex.
\newblock {\em Nature}, 400(6741):233, 1999.

\bibitem[POMT{\etalchar{+}}02]{perez2002oscillations}
Javier Perez-Orive, Ofer Mazor, Glenn~C Turner, Stijn Cassenaer, Rachel~I
  Wilson, and Gilles Laurent.
\newblock Oscillations and sparsening of odor representations in the mushroom
  body.
\newblock {\em Science}, 297(5580):359--365, 2002.

\bibitem[PW14]{polyanskiy2014lecture}
Yury Polyanskiy and Yihong Wu.
\newblock Lecture notes on information theory.
\newblock {\em Lecture Notes for ECE563 (UIUC) and}, 6:2012--2016, 2014.

\bibitem[QKKF08]{quiroga2008sparse}
R~Quian Quiroga, Gabriel Kreiman, Christof Koch, and Itzhak Fried.
\newblock Sparse but not ?grandmother-cell?coding in the medial temporal lobe.
\newblock {\em Trends in cognitive sciences}, 12(3):87--91, 2008.

\bibitem[RP99]{riesenhuber1999hierarchical}
Maximilian Riesenhuber and Tomaso Poggio.
\newblock Hierarchical models of object recognition in cortex.
\newblock {\em Nature neuroscience}, 2(11):1019, 1999.

\bibitem[RPG99]{redgrave1999basal}
Peter Redgrave, Tony~J Prescott, and Kevin Gurney.
\newblock The basal ganglia: a vertebrate solution to the selection problem?
\newblock {\em Neuroscience}, 89(4):1009--1023, 1999.

\bibitem[RV06]{rougier2006emergence}
Nicolas~P Rougier and Julien Vitay.
\newblock Emergence of attention within a neural population.
\newblock {\em Neural Networks}, 19(5):573--581, 2006.

\bibitem[SK08]{smith2008spatial}
Matthew~A Smith and Adam Kohn.
\newblock Spatial and temporal scales of neuronal correlation in primary visual
  cortex.
\newblock {\em Journal of Neuroscience}, 28(48):12591--12603, 2008.

\bibitem[SLW05]{siapas2005prefrontal}
Athanassios~G Siapas, Evgueniy~V Lubenov, and Matthew~A Wilson.
\newblock Prefrontal phase locking to hippocampal theta oscillations.
\newblock {\em Neuron}, 46(1):141--151, 2005.

\bibitem[SN96]{shadlen1996motion}
Michael~N Shadlen and William~T Newsome.
\newblock Motion perception: seeing and deciding.
\newblock {\em Proceedings of the national academy of sciences},
  93(2):628--633, 1996.

\bibitem[SN01]{shadlen2001neural}
Michael~N Shadlen and William~T Newsome.
\newblock Neural basis of a perceptual decision in the parietal cortex (area
  lip) of the rhesus monkey.
\newblock {\em Journal of neurophysiology}, 86(4):1916--1936, 2001.

\bibitem[SS06]{stocker2006noise}
Alan~A Stocker and Eero~P Simoncelli.
\newblock Noise characteristics and prior expectations in human visual speed
  perception.
\newblock {\em Nature neuroscience}, 9(4):578, 2006.

\bibitem[SSB{\etalchar{+}}15]{scholvinck2015cortical}
Marieke~L Sch{\"o}lvinck, Aman~B Saleem, Andrea Benucci, Kenneth~D Harris, and
  Matteo Carandini.
\newblock Cortical state determines global variability and correlations in
  visual cortex.
\newblock {\em Journal of Neuroscience}, 35(1):170--178, 2015.

\bibitem[VRP{\etalchar{+}}18]{verzi2018computing}
Stephen~J Verzi, Fredrick Rothganger, Ojas~D Parekh, Tu-Thach Quach, Nadine~E
  Miner, Craig~M Vineyard, Conrad~D James, and James~B Aimone.
\newblock Computing with spikes: The advantage of fine-grained timing.
\newblock {\em Neural computation}, 30(10):2660--2690, 2018.

\bibitem[wik]{wikiacf}
Activation function.
\newblock \url{https://en.wikipedia.org/wiki/Activation_function}.
\newblock Accessed: 2018-08-08.

\bibitem[Wu17]{YHWu2017lecture}
Yihong Wu.
\newblock Lecture notes on information-theoretic methods for high-dimensional
  statistics.
\newblock {\em Lecture Notes for ECE598YW (UIUC)}, 2017.

\bibitem[YG98]{yuille1998handbook}
AL~Yuille and D~Geiger.
\newblock The handbook of brain theory and neural networks, 1998.

\bibitem[YLS13]{yttri2013lesions}
Eric~A Yttri, Yuqing Liu, and Lawrence~H Snyder.
\newblock Lesions of cortical area lip affect reach onset only when the reach
  is accompanied by a saccade, revealing an active eye--hand coordination
  circuit.
\newblock {\em Proceedings of the National Academy of Sciences},
  110(6):2371--2376, 2013.

\end{thebibliography}

\appendix 

\begin{center}
\Large \bf 
Appendices 
\end{center}

\section{Preliminaries}
\label{app: it}
In this section, we present some preliminaries on information measures and Fano's inequality. Interested readers are referred to \cite{polyanskiy2014lecture} for comprehensive background. 

\subsection{Information Measures} 

Let $X$ and $Y$ be two random variables. The mutual information between $X$ and $Y$, denoted by $I(X; Y)$, measures the dependence between $X$ and $Y$, or, the information about $X$ (resp. $T$) provided by $Y$ (resp. $X$). 
\begin{definition}[Mutual information]
Let $X$ and $Y$ be two random variables. 
\begin{align*}
I(X; Y) : = D(P_{XY} \parallel P_X P_Y),
D(P\parallel Q) & : = \sum_{a\in \calA} P(a) \log \frac{P(a)}{Q(a)},
\end{align*}
where $P_{XY}$ denotes the joint distribution of $X$ and $Y$, and $P_X P_Y$ denotes the product of the marginal distributions of $X$ and $Y$. 
\end{definition}

In the following, we use the notation $X \to Y$ to denote that $Y$ is a (possibly random) function of $X$. Thus, $W\to X\to Y\to \hat{W}$ means that $X$ is a (possibly random) function of $W$; $Y$ is a (possibly random) function of $X$; and $\hat{W}$ is a (possibly random) function of $Y$. 
Fano's inequality:
\begin{theorem}
\cite[Corollary 5.1]{polyanskiy2014lecture} 
 \label{thm: fano}
Let $T: \Theta \to [M]$, and let $\theta \to X\to Y\to \hat{T}(\theta)$ be an arbitrary Markov chain. 
Suppose both $\theta$ and $T(\theta)$ are uniformly distributed over a set of size $M$.  
Then 
\begin{align*}
P_{e} := \prob{T(\theta) \not= \hat{T}(\theta)} 
& \ge  1 - \frac{I(X; Y) + 1}{\log M}. 
\end{align*}
\end{theorem}

\begin{theorem}[Chernoff Bound]
Let $X_1, \cdots, X_n$ be $\iid$ with $X_i\in \sth{0, 1}$ and $\prob{X_1 = 1} = p$. Set $X= \sum_{i=1}^n X_i$. Then 
\begin{itemize}
\item for any $t\in [0, 1-p]$, we have $\prob{X \ge \pth{p+t} n} \le \exp\pth{-n d(p+t \parallel p)}$. 
\item for any $t \in [0, p]$, we have $\prob{X \le \pth{p-t} n} \le \exp\pth{-n d(p-t \parallel p)}$. 
\end{itemize}
\end{theorem}

\section{Proof of Lemma \ref{lm: m Bernoulli kl}}
\label{app: proof of kl lemma}
\begin{proof}[Proof of Lemma \ref{lm: m Bernoulli kl}]
Lemma \ref{lm: m Bernoulli kl} follows easily from the independence between input spike trains and the assumption that the spikes in each input spike train are $\iid$. For completeness, we present the proof as follows. 

Recall 
that 
\begin{align*}
\bm{S} := \qth{\sth{S_t(u_1)}_{t=1}^{T}, \cdots, \sth{S_t(u_n)}_{t=1}^{T}}. 
\end{align*}

Let $\bm{s} = \qth{s_1, \cdots, s_n}$ such that each component $s_i$ is a binary sequence of length $T$, i.e., 
\[
s_i = \qth{b_1^i, \cdots, b_{T}^i} \in \sth{0, 1}^T. 
\]

For each $i=1, \cdots, n$, let $P_{\bm{S}}(\sth{S_t(u_i)}_{t=1}^{T})$ and $Q_{\bm{S}}(\sth{S_t(u_i)}_{t=1}^{T})$ be the marginal distributions of $\sth{S_t(u_i)}_{t=1}^{T}$ under joint distributions $P_{\bm{S}}$ and $Q_{\bm{S}}$ respectively. Similarly, $P_{\bm{S}}(S_t(u_i))$ and $Q_{\bm{S}}(S_t(u_i))$ are the corresponding two marginal distributions of $S_t(u_i)$. 
Thus, we have 
\begin{align*}
&D \pth{ P_{\bm{S}}(\sth{S_t(u_i)}_{t=1}^{T}) \parallel  Q_{\bm{S}}(\sth{S_t(u_i)}_{t=1}^{T})}\\
& \overset{(a)}{=} \sum_{\qth{b_1^i, \cdots, b_{T}^i}} 
P_{\bm{S}}(\sth{S_t(u_i)}_{t=1}^{T} = \qth{b_1^i, \cdots, b_{T}^i}) \log \frac{P_{\bm{S}}(\sth{S_t(u_i)}_{t=1}^{T} = \qth{b_1^i, \cdots, b_{T}^i})}{Q_{\bm{S}}(\sth{S_t(u_i)}_{t=1}^{T} = \qth{b_1^i, \cdots, b_{T}^i})}\\
& \overset{(b)}{=}\sum_{\qth{b_1^i, \cdots, b_{T}^i}} \pth{\prod_{t^{\prime}=0}^{T-1}P_{\bm{S}}(S_{t^{\prime}}(u_i) = b_{t^{\prime}}^i)} \log \frac{\prod_{t=1}^{T}P_{\bm{S}}(S_t(u_i) = b_t^i)}{\prod_{t=1}^{T}Q_{\bm{S}}(S_t(u_i) = b_t^i)}\\
& = \sum_{\qth{b_1^i, \cdots, b_{T}^i}} 
\pth{\prod_{t^{\prime}=0}^{T-1}P_{\bm{S}}(S_{t^{\prime}}(u_i) = b_{t^{\prime}}^i)}  \sum_{t=1}^{T} \log \frac{P_{\bm{S}}(S_t(u_i) = b_t^i)}{Q_{\bm{S}}(S_t(u_i) = b_t^i)}\\
& = \sum_{t=1}^{T} \sum_{\qth{b_1^i, \cdots, b_{T}^i}} \pth{\prod_{t^{\prime}=0}^{T-1}P_{\bm{S}}(S_{t^{\prime}}(u_i) = b_{t^{\prime}}^i)}  \log \frac{P_{\bm{S}}(S_t(u_i) = b_t^i)}{Q_{\bm{S}}(S_t(u_i) = b_t^i)}\\
& = \sum_{t=1}^{T} \sum_{\qth{b_1^i, \cdots, b_{T}^i}} \pth{\prod_{t^{\prime}=0 \& t^{\prime} \not=t}^{T-1}P_{\bm{S}}(S_{t^{\prime}}(u_i) = b_{t^{\prime}}^i)} P_{\bm{S}}(S_{t}(u_i) = b_{t})  \log \frac{P_{\bm{S}}(S_t(u_i) = b_t^i)}{Q_{\bm{S}}(S_t(u_i) = b_t^i)}\\
& \overset{(c)}{=} \sum_{t=1}^{T} \sum_{b_{t}^i} P_{\bm{S}}(S_t(u_i) = b_t^i)  \log \frac{P_{\bm{S}}(S_t(u_i) = b_t^i)}{Q_{\bm{S}}(S_t(u_i) = b_t^i)}\\
& = \sum_{t=1}^{T} \pth{p_i \log \frac{p_i}{q_i} + (1-p_i)\log \frac{1-p_i}{1-q_i}}\\
& = \sum_{t=1}^{T} d(p_i\parallel q_i)  = T \cdot d(p_i\parallel q_i). 
\end{align*}
where $\sum_{\qth{b_1^i, \cdots, b_{T}^i}}$ is the summation over all binary sequences of length $T$. In the last displayed equation, equality (a) follows from the definition of KL divergence; equality (b) is true because of independence of spikes; equality (c) follows from the fact that for any fixed $b_t^i$, 
\[
\sum_{\qth{b_1^i, \cdots, b_{T}^i}\setminus \{t\}} \pth{\prod_{t^{\prime}=1 \& t^{\prime} \not=t}^{T}P_{\bm{S}}(S_{t^{\prime}}(u_i) = b_{t^{\prime}}^i)} =1,
\] 
where we use $\sum_{\qth{b_1^i, \cdots, b_{T}^i}\setminus \{t\}}$ to denote the summation over all binary sequences of length $T$ with the $t$--th entry fixed.

Similarly, we get 
\begin{align*}
D(P_{\bm{S}}\parallel Q_{\bm{S}}) & = \sum_{\bm{s}= \qth{s_1, \cdots, s_n}} P_{\bm{S}}(\bm{S} = \bm{s}) \log \frac{P_{\bm{S}}(\bm{S} = \bm{s})}{Q_{\bm{S}}(\bm{S} = \bm{s})}\\
& = \sum_{i=1}^n D \pth{ P_{\bm{S}}(\sth{S_t(u_i)}_{t=1}^{T}) \parallel  Q_{\bm{S}}(\sth{S_t(u_i)}_{t=1}^{T})}\\
& = \sum_{i=1}^n T d(p_i\parallel q_i)  = T \sum_{i=1}^n d(p_i\parallel q_i),
\end{align*}
proving the lemma. 
\end{proof}

\section{Proof of Lemma \ref{lm: mutual info. k-WTA}}
\label{app:proof of mutual info}
\begin{proof}[Proof of Lemma \ref{lm: mutual info. k-WTA}]
Since mutual information can be viewed as distance to product distributions, by \cite[Theorem 3.4]{polyanskiy2014lecture},  we have 
\begin{align*}
I(X_{\bm{p}}; \bm{S}) ~ & = ~  \min_{Q_{X_{\bm{p}}} Q_{\bm{S}}} D \pth{ P_{X_{\bm{p}}, \bm{S}} \parallel Q_{X_{\bm{p}}} Q_{\bm{S}}}. 
\end{align*}
where $P_{X_{\bm{p}}, \bm{S}}$ is the joint distribution of $X_{\bm{p}}$ and $\bm{S}$, and $Q_{X_{\bm{p}}}$ and $Q_{\bm{S}}$ are any distributions of $X_{\bm{p}}$ and $\bm{S}$, respectively. 

For any fixed $Q_{\bm{S}}$, it holds that 
\begin{align*}
\min_{Q_{X_{\bm{p}}}} D \pth{ P_{X_{\bm{p}}, \bm{S}} \parallel Q_{X_{\bm{p}}} Q_{\bm{S}}}
& = \min_{Q_{X_{\bm{p}}}}  D \pth{ P_{\bm{S} \mid X_{\bm{p}}} P_{X_{\bm{p}}} \parallel Q_{X_{\bm{p}}} Q_{\bm{S}}}\\
& \le D \pth{ P_{\bm{S} \mid X_{\bm{p}}} P_{X_{\bm{p}}} \parallel P_{X_{\bm{p}}} Q_{\bm{S}}},
\end{align*}
where the equality follows from conditioning, and the inequality is true because the best choice over all $Q_{X_{\bm{p}}}$ cannot be worse than any specific choice of $Q_{X_{\bm{p}}}$. Here $\bm{S} \mid X_{\bm{p}}$ denotes the $n$ input spike trains conditioning on the choice of rate assignment. 

For any fixed $Q_{\bm{S}}$, we have 
\begin{align*}
&D \pth{P_{\bm{S} \mid X_{\bm{p}}} P_{X_{\bm{p}}} \parallel P_{X_{\bm{p}}} Q_{\bm{S}} }  
 =  P_{X_{\bm{p}}}(X_{\bm{p}}=\bm{p}^{0}) \sum_{\bm{s}} P_{\bm{S} \mid X_{\bm{p}} = \bm{p}^{0}} (\bm{S} =\bm{s}) \qth{\log \frac{P_{\bm{S} \mid X_{\bm{p}} = \bm{p}^{0}} (\bm{S} =\bm{s}) P_{X_{\bm{p}}}(X_{\bm{p}}=\bm{p}^{0})}{Q_{\bm{S}}(\bm{S}=\bm{s})P_{X_{\bm{p}}}(X_{\bm{p}}=\bm{p}^{0})}} \\
& \qquad \qquad \qquad  + \sum_{i=1}^k \sum_{j=k+1}^n P_{X_{\bm{p}}}(X_{\bm{p}}=\bm{p}^{ij}) \sum_{\bm{s}} P_{\bm{S} \mid X_{\bm{p}} = \bm{p}^{ij}} (\bm{S} =\bm{s}) \qth{\log \frac{P_{\bm{S} \mid X_{\bm{p}} = \bm{p}^{ij}} (\bm{S} =\bm{s}) P_{X_{\bm{p}}}(X_{\bm{p}}=\bm{p}^{ij})}{Q_{\bm{S}}(\bm{S}=\bm{s})P_{X_{\bm{p}}}(X_{\bm{p}}=\bm{p}^{ij})}}\\
&  \qquad \qquad \qquad  \qquad \qquad ~ ~ = \frac{1}{k(n-k)+1} \sum_{\bm{s}} P_{\bm{S} \mid X_{\bm{p}} = \bm{p}^{0}} (\bm{S} =\bm{s}) \qth{\log \frac{P_{\bm{S} \mid X_{\bm{p}} = \bm{p}^{0}} (\bm{S} =\bm{s})}{Q_{\bm{S}}(\bm{S}=\bm{s})}} \\
& \qquad \qquad \qquad \qquad  \qquad  \qquad \quad  +  \frac{1}{k(n-k)+1}\sum_{i=1}^k \sum_{j=k+1}^n \sum_{\bm{s}} P_{\bm{S} \mid X_{\bm{p}} = \bm{p}^{ij}} (\bm{S} =\bm{s}) \qth{\log \frac{P_{\bm{S} \mid X_{\bm{p}} = \bm{p}^{ij}} (\bm{S} =\bm{s})}{Q_{\bm{S}}(\bm{S}=\bm{s})}}\\
&  \qquad \qquad \qquad  \qquad \qquad ~ ~ =  \frac{1}{k(n-k)+1} D \pth{P_{\bm{S} \mid X_{\bm{p}} = \bm{p}^{0}}\parallel Q_{\bm{S}}} 
+ \frac{1}{k(n-k)+1} \sum_{i=1}^k \sum_{j=k+1}^n D \pth{P_{\bm{S} \mid X_{\bm{p}} = \bm{p}^{ij}}\parallel Q_{\bm{S}}}, 
\end{align*}
where $\sum_{\bm{s}}$ is summation over all possible $n$ binary sequences of length $T$. 
Here $P_{\bm{S} \mid X_{\bm{p}} = \bm{p}^{0}}$ is the distribution of $\bm{S}$ with the rate assignment $\bm{p}^{0}$, and 
$P_{\bm{S} \mid X_{\bm{p}} = \bm{p}^{ij}}$ is the distribution of $\bm{S}$ with the rate assignment $\bm{p}^{ij}$.
Choosing $Q_{\bm{S}}$ to be the distribution of $\bm{S}$ with rate assignment $\bm{p}^0$ defined in \eqref{def: rate 0}, then for any $i=1, \cdots, k$ and $j=k+1, \cdots, n$, we have  
\[
D\pth{P_{\bm{S}\mid X_{\bm{p}^{ij}}} \parallel Q_{\bm{S}}} = T(d(r_2\parallel r_1) + d(r_1\parallel r_2)). 
\]
Therefore,
\begin{align*}
I\pth{X_{\bm{p}}\parallel \bm{S}} & \le \frac{1}{k(n-k)+1} \sum_{i=1}^{k} \sum_{j=k+1}^n T(d(r_2\parallel r_1) + d(r_1\parallel r_2)) \\
& \le T(d(r_2\parallel r_1) + d(r_1\parallel r_2)). 
\end{align*}

\end{proof}

\section{Proof of Lemma \ref{lm: quiescent}}
\label{app: lm: quiescent}

By the activation rules in Algorithm \ref{alg: k WTA}, we know that 
\begin{align*}
S_{t_0+m} = 
\begin{cases}
1, ~ \text{if ~ } (b-1)\indc{S_{t_0+m-1}(v_i)=1} + \sum_{r=1}^m \pth{\indc{V_{t_0+m-r} >0} -m\indc{V_{t_0+m-r} \le -1}}>b; \\
0, ~ \text{otherwise.}
\end{cases}
\end{align*}
As all input neurons are quiescent at time $t_0$ and remain to be quiescent for all $t\ge t_0$, it follows that 
\begin{align*}
&(b-1)\indc{S_{t_0+m-1}(v_i)=1} + \sum_{r=1}^m \pth{\indc{V_{t_0+m-r} >0} -m\indc{V_{t_0+m-r} \le -1}}\\
& = (b-1)\indc{S_{t_0+m-1}(v_i)=1} - m \sum_{r=1}^m\indc{V_{t_0+m-r} \le -1}\\
& \le b-1 < b. 
\end{align*}
Thus, $S_{t_0+m}(v_i)=0$ for all $i=1, \cdots, n$. So we have $V_{t_0+m+1}(v_i)=0$ for all $i=1, \cdots, n$, which again implies that $S_{t_0+m+1}(v_i)=0$ for all $i=1, \cdots, n$.  Therefore, we conclude that $S_{t}(v_i)=0$ and $V_t(v_i)=0$ for all $t> t_0+m$.






\section{Proof of Theorem \ref{thm: observation lower bound k_WTA}}
\label{app: proof of theorem lower bound}
\begin{proof}[Proof of Theorem \ref{thm: observation lower bound k_WTA}]
We prove this via a genie-aided argument \cite{jacobs1967lower} by assuming that there is a genie that can access the firing sequences of all the $n$ input neurons. By assuming the existence of a genie, we are essentially considering the centralized setting. 
Clearly, if the error probability is high even in the centralized setting, then no SNNs (which are distributed algorithms) can achieve lower error probability.

Suppose that $T\le \pth{(1-\delta) \log (k(n-k) +1)-1 }T_{\calR}$.  
By \eqref{eq: set R} there exists $r_1, r_2$ such that $r_1\not=r_2$ and 
\begin{align*}
T\le \pth{(1-\delta) \log (k(n-k) +1)-1 }\frac{1}{d(r_2 \parallel r_1) + d(r_1 \parallel r_2)}. 
\end{align*}
Without loss of generality, assume that $r_1 > r_2$. 

Consider the $k(n-k)+1$ possible rate assignments defined in Lemma \ref{lm: mutual info. k-WTA}. Let $\calP$ be the set of such rate assignments. 
By Yao's minimax principle, we know the minimax probability of error is always lower bounded by Bayes probability of error with any prior distribution: 
\[
\max_{\bm{p}\in \calA\calR_k} \prob{\hat{\bm{win}}\pth{\bm{S}} \not= \calW(\bm{p})}  \ge  
\mathbb{E}_{X_{\bm{p}}\sim Unif (\calP) } \qth{\prob{\hat{\bm{win}}\pth{\bm{S}} \not = \calW (X_{\bm{p}})}}, 
\]
where $X_{\bm{p}}\sim Unif (\calP)$ is uniformly distributed over set $\calP$. 
In addition, by Fano's inequality, we have 
\begin{align}
\label{eq: fano k-WTA}
\mathbb{E}_{X_{\bm{p}}\sim Unif (\calP) } \qth{\prob{\hat{\bm{win}}\pth{\bm{S}} \not = \calW(X_{\bm{p}})}}  \ge 1 - \frac{I(X_{\bm{p}}; \bm{S}) + 1}{\log (k(n-k) +1)}. 
\end{align}
Applying Lemma \ref{lm: mutual info. k-WTA}, we get 
\begin{align*}
\max_{\bm{p}\in \calA\calR_k} \prob{\hat{\bm{win}}\pth{\bm{S}} \not= \calW(\bm{p})} 
& \ge 1 - \frac{I(X_{\bm{p}}; \bm{S}) + 1}{\log (k(n-k) +1)}\\
& \ge 1- \frac{T\pth{d(r_2\parallel r_1) + d(r_1\parallel r_2)}  + 1}{\log (k(n-k) +1)}\\
& \ge \delta.
\end{align*}
The last inequality holds as $T \le \pth{(1-\delta) \log (k(n-k) +1)-1 }T_{\calR}$. 
\end{proof}

\section{Proof of Theorem \ref{thm: k wta threhold real}}
\label{app: proof of algorithm}
The proof of Theorem \ref{thm: k wta threhold real} uses the following technical fact and lemma. 
\begin{fact}
\label{fact: monotone}
For any given $p\in (0, 1)$ and $b>0$, let $f_{p, b}: \reals \to \reals$, defined as: for all $t>0$,  
\begin{align*}
f_{p, b}(t) : = \exp\pth{-t d\pth{\frac{b}{t} \parallel p}}.  
\end{align*}
Function $f_{p, b}(\cdot) $ is increasing when $t\in (0, \frac{b}{p})$ and decreasing when $t\ge \frac{b}{p}$.
\end{fact}
This fact follows immediately from a simple algebra. 
\begin{lemma}
\label{lm: bernoulli int}
Assume $u, v\in [c, C] \subseteq (0,1)$. Then for any $\alpha \in (0, 1)$, 
\begin{align*}
d\pth{(1-\alpha)u+\alpha v \parallel u } \ge \frac{\alpha^2 c(1-C)}{2C(1-c)} \,\pth{d\pth{u\parallel v} + d\pth{v\parallel u}}.  
\end{align*}
\end{lemma}
\begin{proof}
Note that for any fixed $q\in [c, C]$, $d\pth{x \parallel q}$ is a function of $x$, where $x\in [c, C]$. 
In addition, by simple algebra, we have 
\begin{align}
\label{eq: kl derivative}
d^{\prime}\pth{x \parallel q}  = \log \frac{(1-q)x}{q(1-x)}, ~ \text{and} ~ d^{\prime \prime}\pth{x \parallel q}  = \frac{1}{x(1-x)}. 
\end{align}

By Taylor expansion, we have 
\begin{align*}
d\pth{(1-\alpha)u+\alpha v \parallel u } & = d\pth{u \parallel u } + \pth{(1-\alpha)u+\alpha v - u} d^{\prime}\pth{u \parallel u } \\
& \quad +  \frac{\pth{(1-\alpha)u+\alpha v - u}^2}{2} d^{\prime\prime}\pth{\xi \parallel u}, 
\end{align*}
where $\xi \in \qth{\min \{u, (1-\alpha)u+\alpha v\}, ~ \max \{u, (1-\alpha)u+\alpha v\}}$. By \eqref{eq: kl derivative}, 
\begin{align*}
d\pth{(1-\alpha)u+\alpha v \parallel u } & = 0+ 0 + \frac{1}{\xi(1-\xi)} \frac{\alpha^2(u-v)^2}{2} \ge  \frac{\alpha^2(u-v)^2}{2C(1-c)} . 
\end{align*}

On the other hand, since $d\pth{u\parallel v} + d\pth{u\parallel v}$ is symmetric in $u$ and $v$, without loss of generality, assume that $u\ge v$. 
We have  
\begin{align*}
d\pth{u\parallel v} + d\pth{u\parallel v}& = (u-v) \log \frac{u(1-v)}{v(1-u)}\\
&= (u-v) \log \pth{1+\frac{u-v}{v(1-u)}}\\
& \le (u-v) \frac{u-v}{v(1-u)} = \frac{(u-v)^2}{v(1-u)} \le  \frac{(u-v)^2}{c(1-C)}\\
& \le \frac{2C(1-c)}{c(1-C)\alpha^2}  d\pth{(1-\alpha)u+\alpha v \parallel u }, 
\end{align*}
proving the lemma. 
\end{proof}

Now we are ready to prove Theorem \ref{thm: k wta threhold real}. 
\begin{proof}[{\bf Proof of Theorem \ref{thm: k wta threhold real}}]
Without loss of generality, assume that 
\[
p_1 \ge \cdots \ge p_k > p_{k+1} \ge \cdots \ge p_{n}. 
\]

For a given rate assignment $\bm{p}\in \calA \calR$,  
define $\tau_1, \tau_2, \cdots, \tau_n$ as 
\begin{align*}
\tau_i :=   \inf_{t}\sth{t: ~ \sum_{r=0}^{\min\{t, m^*\} } S_r(u_i) \ge b},  ~~ \forall ~ i =1, \cdots, n. 
\end{align*}
To show Theorem \ref{thm: k wta threhold real}, it is enough to show that with probability $1-\delta$, 
\begin{align}
\tau_i &< \tau_j ~~ \forall ~ i=1, \cdots, k, ~ \text{and}~ j=k+1, \cdots, n; \label{eq: true}\\
\text{and} ~ ~  
\tau_i &\le m^* ~~ \forall ~ i=1, \cdots, k. \label{eq: hitting}. 
\end{align}

Before diving into proving \eqref{eq: true} and \eqref{eq: hitting} hold with probability at least $1-\delta$, let's check the sufficiency of \eqref{eq: true} and \eqref{eq: hitting}. 
Let $t_0:= \max_{1\le i \le k} \tau_i$.
Let $\calE$ be the event on which \eqref{eq: true} and \eqref{eq: hitting} hold. 
Clearly, conditioning on event $\calE$, we have  
\[
\pth{\max_{1\le i \le k} \tau_i ~\mid \calE} = ~ t_0 \mid \calE \le m^*-1 \le m-1,  
\]
and  
\[
\pth{\max_{1\le i \le k} \tau_i ~ \mid ~ \calE} = t_0 \mid \calE ~ < ~ \tau_j \mid \calE ~ ~ \forall j=k+1, \cdots, n.  
\]

Notably, for any $t\le t_0 \le m -1$ and for $i =1, \cdots, n$, 
\begin{align*}
\qth{\sum_{r=1}^{t} \indc{V_{r}(v_i) >0} - m \sum_{r=1}^{t} \indc{V_{r}(v_i) \le -1}}_+ 
& \le \sum_{r=1}^{t} \indc{V_{r}(v_i) >0}  \le \sum_{r=1}^{t} S_r(u_i). 
\end{align*}
Thus, conditioning on $\calE$, at most $k-1$ output neurons ever spike by time $t_0$. So we have (1) $\indc{V_t(v_i) \le -1} =0$, and (2) $\indc{V_t(v_i) >0}=S_t(u_i)$, for all $i=1, \cdots, n$ and for all $t\le t_0$. 
In addition, we have for all $t\le t_0$, 
\begin{align*}
&(b-1) \indc{S_{t}(v_i)=1} + \qth{\sum_{r=1}^{t} \indc{V_r(v_i)>0} - m\sum_{r=1}^{t} \indc{V_r(v_i)\le -1}}_+  \\
&= (b-1) \indc{S_{t}(v_i)=1} + \sum_{r=1}^{t} \indc{V_r(v_i)>0} \\
&= (b-1) \indc{S_{t}(v_i)=1} + \sum_{r=1}^{t} S_r(u_i). 
\end{align*}
By the activation rules in Algorithm \ref{alg: k WTA}, we know, conditioning on $\calE$, at time $t_0+1\le m^*$, output neurons $v_1, \cdots, v_k$ spike simultaneously, and output neurons $v_{k+1}, \cdots, v_n$ do not spike, proving (1) in Theorem \ref{thm: k wta threhold real}. 
By the choice of $t_0$, we know that, on $\calE$, $t_0+1$ is the first time that $k$ output neurons spike simultaneously, and no other $k$ output neurons ever spike simultaneously, proving (2) in Theorem \ref{thm: k wta threhold real}. 

By a simple induction argument, it can be shown that conditioning on $\calE$, in each of the time slot $t$ such that $t_0+1 \le t \le m+1$, output neurons $v_1, \cdots, v_k$ spike, and no other output neurons (i.e., output neurons $v_{k+1}, \cdots, v_n$ do not spike). 
Let's consider the case when $t=(m+1)+1$. As among output neurons, only $v_1, \cdots, v_k$ spike, and no other output neurons spike for any $t^{\prime} \le m+1$, it follows that 
\[
m\sum_{r=1}^{m} \indc{V_{t-r}(v_i)\le -1} =0, ~~ \forall ~  v_1, \cdots, v_k. 
\] 
Thus, for these $k$ output neurons, 
\begin{align*}
&(b-1) \indc{S_{t-1}(v_i)=1} + \qth{\sum_{r=1}^{m} \indc{V_{t-r}(v_i)>0} - m\sum_{r=1}^{m} \indc{V_{t-r}(v_i)\le -1}}_+ \\
& = (b-1) + \sum_{r=1}^{m} \indc{V_{t-r}(v_i)>0} \\
& =  (b-1) + \sum_{r=1}^{m} \indc{V_{t-1-r}(v_i)>0}  + \indc{V_{t-1}(v_i)>0} - \indc{V_{t-1-m}(v_i)>0}\\
& \ge b-2  + \sum_{r=1}^{m} \indc{V_{t-1-r}(v_i)>0} \\
& = b-2 + \sum_{r=1}^m S_{r}(u_i) \ge 2b-2 \ge b,
\end{align*}
where the last inequality holds as long as $b\ge 2$. For output neurons $v_{k+1}, \cdots, v_n$, we have \begin{align*}
&(b-1) \indc{S_{t-1}(v_i)=1} + \qth{\sum_{r=1}^{m} \indc{V_{t-r}(v_i)>0} - m\sum_{r=1}^{m} \indc{V_{t-r}(v_i)\le -1}}_+ \\
& \le \sum_{r=1}^{m} \indc{V_{t-r}(v_i)>0} \\
& \overset{(a)}{=} \sum_{r=1}^{m} \indc{V_{t-1-r}(v_i)>0} + \indc{V_{t-1}(v_i)>0} - \indc{V_{t-1-m}(v_i)>0}\\
& = \sum_{r=1}^{m} \indc{V_{t-1-r}(v_i)>0} - \indc{V_{t-1-m}(v_i)>0} \\
& \le \sum_{r=1}^{m} \indc{V_{t-1-r}(v_i)>0}
=\sum_{r=1}^{m} \indc{V_{r}(v_i)>0} < b. 
\end{align*}
Equality (a) follows because at time $t-1$, output neurons $v_1, \cdots, v_k$ spike, resulting in $\indc{V_{t-1-r}(v_i)>0}=0$ for $i\not=1, \cdots, k$.  
Thus, we know conditioning on event $\calE$, at time $(m+1)+1$, the output neurons $v_1, \cdots, v_k$ spike, and no other output neuron spike. It can be shown by a simple induction that at each time $t$ such that $t_0+1 \le t \le m+b$, the output neurons $v_1, \cdots, v_k$ spike, and no other output neurons spike. This proves (3) in Theorem \ref{thm: k wta threhold real}. 

\vskip \baselineskip 

Next we prove \eqref{eq: true} and \eqref{eq: hitting}. 
By definition of $\tau_j$, we know that $\tau_j \le m^*$ for all $j=1, \cdots, n$. Thus, we only need to show that 
with probability $1-\delta$, 
\[
\tau_i < \tau_j ~~ \forall ~ i=1, \cdots, k, ~ \text{and}~ j=k+1, \cdots, n, 
\]
which is the focus of the remainder of our proof.

\vskip \baselineskip

Note that 
\begin{align}
\label{eq: k-wta prob error}
\nonumber
& \prob{\tau_i < \tau_j, ~~\forall i\in \{1, \cdots, k\}, \forall j\in \{k+1, \cdots, n\}} \\
\nonumber
& = \prob{\tau_i < \tau_j, \& ~ \tau_i < m^*, ~~\forall i\in \{1, \cdots, k\},  \forall j\in \{k+1, \cdots, n\}} \\
& \ge 1- \sum_{i=1}^k \sum_{j=k+1}^n \prob{\tau_i \ge \tau_j, ~ \text{or}~ \tau_i=m^*}. 
\end{align}
For each term in the summation of \eqref{eq: k-wta prob error}, we have 
\begin{align}
\label{eq: k-wta whole term error}
\prob{\tau_i \ge \tau_j, ~ \text{or}~ \tau_i=m^*} = \prob{\tau_i=m^*} + \prob{\tau_i \ge \tau_j, \& ~ \tau_i<m^*},
\end{align}
which follows from the fact that $\prob{A\cup B} = \prob{A} + \prob{B-A}$ for any sets $A$ and $B$. 
Note that $m^*p_i \ge b$. By Chernoff bound, the first term in \eqref{eq: k-wta whole term error} is bounded as
\begin{align}
\label{eq: k-wta error probability}
\prob{\tau_i =m^*} = \prob{\sum_{r=0}^{m^*}S_r(u_i)\le b} \le \exp\pth{-m^* \cdot d\pth{\frac{b}{m^*}\parallel p_i}}. 
\end{align}
For the second term in \eqref{eq: k-wta whole term error}, we have 
\begin{align*}
\prob{\tau_i\ge \tau_j ~~\text{and}~ \tau_i<m^*} & = \prob{\sum_{r=0}^{\tau_i} S_r(u_j) \ge b, ~~\text{and} ~ \tau_i <m^*} \\
& \le \exp\pth{-t^* \cdot d\pth{\frac{b}{t^*} \parallel p_{k+1}}} + \exp\pth{-t^* \cdot d\pth{\frac{b}{t^*} \parallel p_{k}}}, 
\end{align*}
where $t^* \in \pth{\frac{b}{p_{k+1}}, \frac{b}{p_k}}$. 
Thus, \eqref{eq: k-wta whole term error} is upper bounded as 
\begin{align*}
\prob{\tau_i \ge \tau_j, ~ \text{or}~ \tau_i=m^*} & \le \exp\pth{-m^* \cdot d\pth{\frac{b}{m^*} \parallel p_{k+1}}} \\
& \quad + \exp\pth{-t^* \cdot d\pth{\frac{b}{t^*} \parallel p_{k+1}}} + \exp\pth{-t^* \cdot d\pth{\frac{b}{t^*} \parallel p_{k}}}\\
& \le \exp\pth{-t^* \cdot d\pth{\frac{b}{t^*} \parallel p_{k+1}}} + 2\exp\pth{-t^* \cdot d\pth{\frac{b}{t^*} \parallel p_{k}}}. 
\end{align*}
Eq \eqref{eq: k-wta prob error} is bounded as 
\begin{align*}
& \prob{\tau_i < \tau_j, ~~\forall i\in \{1, \cdots, k\}, \forall j\in \{k+1, \cdots, n\}}\\
& \ge 1- \sum_{i=1}^k \sum_{j=k+1}^n \prob{\tau_i \ge \tau_j, ~ \text{or}~ \tau_i=m^*} \\
& \ge 1-\sum_{i=1}^k \sum_{j=k+1}^n \pth{\exp\pth{-t^* \cdot d\pth{\frac{b}{t^*} \parallel p_{k+1}}} + 2\exp\pth{-t^* \cdot d\pth{\frac{b}{t^*} \parallel p_{k}}}}\\
& = 1-k(n-k) \pth{\exp\pth{-t^* \cdot d\pth{\frac{b}{t^*} \parallel p_{k+1}}} + 2\exp\pth{-t^* \cdot d\pth{\frac{b}{t^*} \parallel p_{k}}}}. 
\end{align*}

Let $t^* = \frac{b}{(p_k+p_{k+1})/2}$, it holds that 
\begin{align*}
\exp\pth{-t^* \cdot d\pth{\frac{b}{t^*} \parallel p_{k+1}}} & = \exp\pth{- \frac{b}{(p_k+p_{k+1})/2} \cdot d\pth{\frac{p_k+p_{k+1}}{2} \parallel p_{k+1}}},\\
2\exp\pth{-t^* \cdot d\pth{\frac{b}{t^*} \parallel p_{k}}} & = 2\exp\pth{- \frac{b}{(p_k+p_{k+1})/2} \cdot d\pth{\frac{p_k+p_{k+1}}{2} \parallel p_{k}}}. 
\end{align*}
By Lemma \ref{lm: bernoulli int}, we know 
\begin{align*}
d\pth{\frac{p_k+p_{k+1}}{2} \parallel p_{k+1}} \ge \frac{c(1-C)}{8C(1-c)} \pth{d(p_{k+1} \parallel p_k) + d(p_k \parallel p_{k+1})},  
\end{align*}
and, 
\[
d\pth{\frac{p_k+p_{k+1}}{2} \parallel p_{k}} \ge  \frac{c(1-C)}{8C(1-c)} \pth{d(p_{k+1} \parallel p_k) + d(p_k \parallel p_{k+1})}. 
\]
Thus, we get 
\begin{align*}
&\prob{\tau_i < \tau_j, ~~\forall i\in \{1, \cdots, k\}, \forall j\in \{k+1, \cdots, n\}} \\
& \ge 1-3k(n-k)\exp\pth{-\frac{2b}{p_k+p_{k+1}} \frac{c(1-C)}{8C(1-c)}\pth{d(p_k\parallel p_{k+1}) +d(p_{k+1}\parallel p_{k})}} 
\end{align*}
Since $b= \frac{8C^2(1-c)}{c(1-C)}\pth{\log \frac{3}{\delta} + \log k(n-k)} T_{\calR}$, we have 
\[
3k(n-k)\exp\pth{-\frac{2b}{p_k+p_{k+1}} \frac{c(1-C)}{8C(1-c)}\pth{d(p_k\parallel p_{k+1}) +d(p_{k+1}\parallel p_{k})}} \le \delta.
\]
Thus, $ \prob{\tau_i < \tau_j, ~~\forall i\in \{1, \cdots, k\}, \forall j\in \{k+1, \cdots, n\}} \le 1-\delta$.

In addition, 
\begin{align*}
t^* = \frac{2b}{p_{k}+p_{k+1}} \le \frac{1}{c}b = m^* \le m, 
\end{align*}
completing the proof of Theorem \ref{thm: k wta threhold real}. 
\end{proof}

%
%

\end{document}